\documentclass[acmsmall,10pt]{acmart}

\usepackage{booktabs}   
\usepackage{subcaption} 

\setcopyright{rightsretained}
\acmPrice{}
\acmDOI{10.1145/3236782}
\acmYear{2018}
\copyrightyear{2018}
\acmJournal{PACMPL}
\acmVolume{2}
\acmNumber{ICFP}
\acmArticle{87}
\acmMonth{9}

\usepackage{amsmath}  
\usepackage{amssymb}
\usepackage{mathpartir}
\usepackage[figuresright]{rotating}
\usepackage{array}
\usepackage{mmm}  

\usepackage{bigints}

\newtheorem{property}[section]{Property}

\newmeta{\c}{\itop{v}}
\newmeta{\ty}{\otype}

\newrmop{\Eval}{Eval}

\newcommand{\otype}{\ensuremath{\alpha}\xspace}
\newcommand{\op}{\itop{op}}
\newcommand{\deltafun}{\ensuremath{\delta}}

\def\sfop#1{\ma{\mathop{\hbox{\sf #1\/}}\nolimits}}

\renewcommand\newttop[1]{\ensuremath{\ttop{#1}}\xspace}
\newcommand\logop{\newttop{log}}
\newcommand\expop{\newttop{exp}}
\newcommand\realop{\newttop{real?}}

\newcommand{\Exptype}{\sfop{exp}\xspace}
\newcommand{\Conttype}{\sfop{cont}\xspace}
\newcommand{\conttype}[1]{\Conttype}

\newcommand{\Valtype}{\sfop{val}\xspace}

\newcommand{\Ent}{\H}
\renewcommand{\H}{\ensuremath{\mathbb{S}}\xspace}
\renewcommand{\cons}{\ensuremath{{\bfop{::}}}}
\newcommand{\R}{\ensuremath{\mathbb{R}}\xspace}
\newcommand{\Rplus}{\ensuremath{\mathbb{R}^+}\xspace}

\newcommand{\nextstep}{\ensuremath{\itop{next-state}}\xspace}

\newcommand\appexp[2]{\ensuremath{\texttt({#1}\ {#2} \texttt) }}
\newcommand\lambdaexp[2]{\ensuremath{\lambda{#1}.{#2}}}
\newcommand\letexp[3]{\ensuremath{\ttop{let}{#1}={#2}\,\ttop{in}\,{#3}}}
\newcommand\ifexp[3]{\ensuremath{\ttop{if}\,{#1}\,\ttop{then}\,{#2}\,\ttop{else}\,{#3}}}
\newcommand\opexp[2]{\ensuremath{#1\texttt{(}#2\texttt{)}}}
\newcommand\opexptt[2]{\ensuremath{\texttt{#1(}#2\texttt{)}}}

\newcommand\parenthexp[1]{\texttt{(\!}#1\texttt{)}}
\newcommand\lambdaexpp[2]{\ensuremath{\texttt{(}\lambda{#1}.{#2}\texttt{)}}}
\newcommand\letexpp[3]{\ensuremath{\parenthexp{\ttop{let}{#1}={#2}\,\ttop{in}\,{#3}}}}

\newcommand\const[1]{\ensuremath{\ttop{c}_{#1}}}
\newcommand\sampleexp{\ensuremath{\ttop{sample}}\xspace}
\newcommand\factorexp[1]{\ensuremath{\ttop{factor}\,{#1}}}

\newcommand\haltcont{\ensuremath{\ttop{halt}}}

\newcommand\letcont[3]{\ensuremath{({#1} \to {#2}){#3}}\xspace}

\newcommand\hastype[3]{\ensuremath{{#1}{\,\vdash\,}{#2}\ {#3}}}

\newcommand{\Config}{\ensuremath{\rmop{Config}}\xspace}
\newcommand{\finishcomp}{\ensuremath{\itop{finish}}\xspace}

\newcommand\csep{~\vert~}

\newcommand\configsw[2]{\nconfig{\s}{#1}{#2}{\t}{w}}

\newcommand\nconfig[5]{\ensuremath{\langle {#1} \csep {#2} \csep {#3}
    \csep {#4} \csep {#5} \rangle}\xspace}

\newcommand\ev[4]{\ensuremath{\rmop{ev}({#1},{#2},{#3},{#4})}}
\newcommand\ew[4]{\ensuremath{\rmop{ew}({#1},{#2},{#3},{#4})}}

\newcommand\eval[6]{\ensuremath{\rmop{eval}({#1},{#2},{#3},{#4},{#5}, {#6})}}
\newcommand\aeval[7]{\ensuremath{\rmop{eval}^{({#1})}({#2},{#3},{#4},{#5},{#6},{#7})}}

\newcommand\smint[2]{\textstyle \int {#1} \;{#2}}
\newcommand\smiint[3]{\textstyle \iint {#1} \;{#2}\,{#3}}
\newcommand\smiiint[4]{\textstyle \iiint {#1} \;{#2}\,{#3}\,{#4}}
\newcommand\smiiiint[5]{\textstyle \iiiint {#1} \;{#2}\,{#3}\,{#4}\,{#5}}

\newcommand\smintBB[4]{\textstyle \int_{#1}^{#2} {#3} \;{#4}}

\newcommand\leb[2]{\smint{#1}{d#2}}
\newcommand\dint[3]{\smiint{#1}{d#2}{d#3}}
\newcommand\tint[4]{\smiiint{#1}{d#2}{d#3}{d#4}}
\newcommand\qint[5]{\smiiiint{#1}{d#2}{d#3}{d#4}{d#5}}
\newcommand{\lebs}[1]{\leb {#1}{\s}}

\newcommand{\iterint}[3]{\leb{\dots \leb{#1}{#2} \dots \,}{#3}}

\newcommand\Meas[3]{\ensuremath{\mu({#1},{#2},{#3})\xspace}}
\newcommand\aMeas[4]{\ensuremath{\mu^{(#1)}({#2},{#3},{#4})\xspace}}

\newcommand\vmeasM[1]{\ensuremath{\hat{\mu}}(#1, -)}
\newcommand\vmeas[2]{\ensuremath{\hat{\mu}}(#1, #2)}
\newcommand\vmeasfun{\ensuremath{\hat{\mu}}}

\newcommand\logrel[5]{\ensuremath{({#4},{#5}) \in
    #1}^{#2}_{#3}\xspace}

\newcommand\valrel[4]{\logrel{\bbord V}{#1}{#2}{#3}{#4}}
\newcommand\exprel[4]{\logrel{\bbord E}{#1}{#2}{#3}{#4}}
\newcommand\contrel[4]{\logrel{\bbord K}{#1}{#2}{#3}{#4}}

\newcommand\ilogrel[5]{\ensuremath{({#4},{#5}) \in #1}^{#2}_{#3}\xspace}

\newcommand\ivalrel[4]{\ilogrel{\bbord V}{#1}{#2}{#3}{#4}}
\newcommand\iexprel[4]{\ilogrel{\bbord E}{#1}{#2}{#3}{#4}}
\newcommand\icontrel[4]{\ilogrel{\bbord K}{#1}{#2}{#3}{#4}}
\newcommand\iciurel[4]{\logrel{\bbord {CIU}}{#1}{#2}{#3}{#4}}
\newcommand\iarbrel[4]{\logrel{R}{#1}{#2}{#3}{#4}}

\newcommand\ciurel[4]{\logrel{\bbord {CIU}}{#1}{#2}{#3}{#4}}
\newcommand\ctxrel[4]{\logrel{\bbord {CTX}}{#1}{#2}{#3}{#4}}
\newcommand\ciurelset[2]{\logrelset{\bbord {CIU}}{#1}{#2}}
\newcommand\ctxrelset[2]{\logrelset{\bbord {CTX}}{#1}{#2}}

\newcommand\ctxeq{\ensuremath{=_{\text{\scriptsize{ctx}}}}\xspace}

\newcommand\letvax{\ensuremath{\rmop{let}_v}\xspace}
\newcommand\letidax{\ensuremath{\rmop{let}_{\itop{\scriptsize{id}}}}\xspace}
\newcommand\letsax{\ensuremath{\rmop{let}_S}\xspace}

\newcommand\ityenvrel[4]{\ilogrel{\bbord G}{#2}{#1}{#3}{#4}}

\newcommand\logrelset[3]{\ensuremath{#1}^{#2}_{#3}\xspace}

\newcommand\ivalrelset[2]{\logrelset{\bbord V}{#1}{#2}}
\newcommand\iexprelset[2]{\logrelset{\bbord E}{#1}{#2}}
\newcommand\icontrelset[2]{\logrelset{\bbord K}{#1}{#2}}
\newcommand\ityenvrelset[2]{\logrelset{\bbord G}{#1}{#2}}
\newcommand\iciurelset[2]{\logrelset{\bbord {CIU}}{#1}{#2}}
\newcommand\ictxrelset[2]{\logrelset{\bbord {CTX}}{#1}{#2}}
\newcommand\iarbrelset[2]{\logrelset{R}{#1}{#2}}

\newcommand\icontrelalt[4]{({#3},{#4}) \in \icontrelset{#1}{#2}}

\newcommand\valrelset[2]{\logrelset{\bbord V}{#1}{#2}}
\newcommand\exprelset[2]{\logrelset{\bbord E}{#1}{#2}}
\newcommand\contrelset[2]{\logrelset{\bbord K}{#1}{#2}}

\newcommand\wider{\renewcommand\arraystretch{1.5}}

\newrmop{\dom}{dom}

\newitop{\strict}{strict}
\newrmop{\lfp}{lfp}

\newcommand{\bigstep}[4]{\ensuremath{#1} \proves {#2} \Downarrow {#3},{#4}}

\newcommand\sem[1]{\ensuremath{\lbrack\!\lbrack{#1}\rbrack\!\rbrack}}

\newitop{\Measop}{Meas}

\newcommand\proves{\ensuremath{\vdash}}

\newitop{\LL}{LL}
\newitop{\RR}{RR}

\newcommand{\mathblock}[1]{{\renewcommand{\arraystretch}{1}\begin{array}[t]{@{}l} #1 \end{array}}}

\newitop{\return}{return}

\newitop{\bindop}{bind}

\newitop{\invcdf}{invcdf}

\newdom{\Bool}{Bool}
\newitop{\true}{true}
\newitop{\false}{false}

\newdom{\Unit}{\mbox{\upshape()}}  

\newitop{\Real}{Real}

\newitop{\prior}{prior}
\newitop{\model}{model}

\newitop{\bayes}{bayes}
\newitop{\post}{post}

\newcommand{\just}[1]{\quad \hbox{(#1)}}

\newcommand{\comdots}{, \dots,}

\newcommand\SigmaReal{\ensuremath{\Sigma_{\mathbb{R}}}}
\newcommand\SigmaR{\ensuremath{\Sigma_{\mathbb{R}}}}

\newcommand\proofinappendix{\begin{proof}See \theappendixref.\end{proof}}
\newcommand\theappendixref{Appendix~\ref{sec:appendix-proofs}}

\newitop{\exp}{e}
\newmeta{\e}{\itop{e}}
\newmeta{\c}{\itop{v}}
\newmeta{\ty}{\otype}

\newenvironment{Case}[1]%
   {\textbf{Case\ }{#1}\,\textbf{:}}
   {}

\newcommand\tophantom{\ensuremath{\to^{\hphantom{*}}}}

\begin{document}
\title{Contextual Equivalence for a Probabilistic Language with
  Continuous Random Variables and Recursion}
\thanks{This paper is an extended version of \citet{ICFPVersion} (ICFP 2018).}

\author{Mitchell Wand}
\affiliation{
  \department{College of Computer and Information Science}
  \institution{Northeastern University}
  \streetaddress{360 Huntington Ave, Room 202WVH}
  \city{Boston}
  \state{MA}
  \country{USA}
  \postcode{02115}}
\email{wand@ccs.neu.edu}

\author{Ryan Culpepper}
\affiliation{
  \department{Faculty of Information Technology}
  \institution{Czech Technical University in Prague}
  \streetaddress{Th\'akurova 9}
  \city{Prague}
  \postcode{16000}
  \country{Czech Republic}
}
\email{ryanc@ccs.neu.edu}

\author{Theophilos Giannakopoulos}
\affiliation{
  \department{FAST Labs}
  \institution{BAE Systems}
  \streetaddress{600 District Ave}
  \city{Burlington}
  \state{MA}
  \country{USA}
  \postcode{01803}
}
\email{tgiannak@alum.wpi.edu}

\author{Andrew Cobb}
\affiliation{
  \department{College of Computer and Information Science}
  \institution{Northeastern University}
  \streetaddress{360 Huntington Ave, Room 202WVH}
  \city{Boston}
  \state{MA}
  \country{USA}
  \postcode{02115}}
\email{acobb@ccs.neu.edu}

\begin{abstract}
  We present a complete reasoning principle for contextual equivalence
  in an untyped probabilistic language. The language includes
  continuous (real-valued) random variables, conditionals, and
  scoring. It also includes recursion, since the standard
  call-by-value fixpoint combinator is expressible.

  We demonstrate the usability of our characterization by proving
  several equivalence schemas, including familiar facts from lambda
  calculus as well as results specific to probabilistic programming.
  In particular, we use it to prove that reordering the random draws
  in a probabilistic program preserves contextual equivalence. This
  allows us to show, for example, that
  \begin{displaymath}
    ({\letexp{x}{e_1}{\letexp{y}{e_2}{e_0}}}) \ctxeq
    ({\letexp{y}{e_2}{\letexp{x}{e_1}{e_0}}})
  \end{displaymath}
  (provided $x$ does not occur free in $e_2$ and $y$ does not occur
  free in $e_1$) despite the fact that $e_1$ and $e_2$ may have
  sampling and scoring effects.
\end{abstract}




\maketitle

\section{Introduction} \label{sec:probprog}


A \emph{probabilistic programming language} is a programming language
enriched with two features---\emph{sampling} and \emph{scoring}---that
enable it to represent probabilistic models. We introduce these two
features with an example program that models linear regression.

The first feature, \emph{sampling}, introduces probabilistic
nondeterminism. It is used to represent random variables. For example,
let $\opexptt{normal}{m,s}$ be defined to nondeterministically produce
a real number distributed according to a normal (Gaussian)
distribution with mean $m$ and scale $s$.

Here is a little model of linear regression that uses \texttt{normal}
to randomly pick a slope and intercept for a line and then defines
\texttt{f} as the resulting linear function:
\begin{alltt}
  A = normal(0, 10)
  B = normal(0, 10)
  f(x) = A*x + B
\end{alltt}
This program defines a distribution on lines, centered on
$y = 0x+0$, with high variance.  This distribution is called the
\emph{prior}, since it is specified prior to considering any
evidence.

The second feature, \emph{scoring}, adjusts the likelihood of the
current execution's random choices. It is used to represent
conditioning on observed data.

Suppose we have the following data points: $\{ (2.0, 2.4), (3.0, 2.7),
(4.0, 3.0) \}$.  The smaller the error between the result of \texttt{f}
and the observed data, the better the choice of \texttt{A} and
\texttt{B}. We express these observations with the following addition
to our program:
\begin{alltt}
  factor normalpdf(f(2.0)-2.4; 0, 1)
  factor normalpdf(f(3.0)-2.7; 0, 1)
  factor normalpdf(f(4.0)-3.0; 0, 1)
\end{alltt}
Here scoring is performed by the \texttt{factor}
form, which takes a positive real number to multiply into the current
execution's likelihood.  We use $\opexptt{normalpdf}{\_; \texttt{0,1}}$---the
density function of the standard normal distribution---to convert the
difference between predicted and observed values into the score. This
scoring function assigns high likelihood when the error is near 0,
dropping off smoothly to low likelihood for larger errors.

After incorporating the observations, the program defines a
distribution centered near $y = 0.3x + 1.8$, with low variance.  This
distribution is often called the \emph{posterior distribution}, since it represents
the distribution after the incorporation of evidence.

Computing the posterior distribution---or a workable approximation
thereof---is the task of \emph{probabilistic inference}.  We say that
this program has \emph{inferred} (or sometimes \emph{learned}) the
parameters \texttt{A} and \texttt{B} from the data. Probabilistic
inference encompasses an arsenal of techniques of varied applicability
and efficiency. Some inference techniques may benefit if the program
above is transformed to the following shape:
\begin{alltt}
  A = normal(0, 10)
  factor \(Z\)(A)
  B = normal(\(M\)(A), \(S\)(A))
\end{alltt}
The transformation relies on the conjugacy relationship between the
normal prior for \texttt{B} and the normal scoring function of the
observations. A useful equational theory for probabilistic programming
must incorporate facts from mathematics in addition to standard
concerns such as function inlining.


In this paper we build a foundation for such an equational theory for
a probabilistic programming language. In particular, our language
supports
\begin{itemize}
\item sampling continuous random variables,
\item scoring (soft constraints), and
\item conditionals, higher-order functions, and recursion.
\end{itemize}
Other such languages include Church~\citep{Church}, its descendants
such as Venture~\citep{Venture} and Anglican~\citep{Anglican}, and
other languages~\citep{Hansei,Hakaru,ProbC}
and language models~\citep{Park08,Huang:2016,Borgstrom:2016}.
Our framework is able to justify the transformation above.

In Section~\ref{sec:lang} we present our model of a probabilistic
language, including its syntax and semantics.
We then define our logical relation (Section~\ref{sec:log-rel}), our
CIU relation (Section~\ref{sec:ciu}), and contextual ordering
(Section~\ref{sec:ciu-ctx}); and we prove that all three relations
coincide.
As usual, contextual ordering is powerful but difficult to prove
directly.  The virtue of the logical relation is that it eliminates
the need to reason about arbitrary syntactic contexts; they are boiled
down to their essential components: substitutions and continuations
(evaluation contexts). The CIU relation~\cite{MasonTalcott:1991} is a
further simplification of the logical relation; it offers the easiest
way to prove relationships between specific terms.
In Section~\ref{sec:equivs} we define contextual equivalence and use
the CIU relation to demonstrate a catalog of useful equivalence
schemas, including $\beta_v$ and \ttop{let}-associativity, as well as
a method for importing first-order equivalences from mathematics.
One unusual equivalence is \ttop{let}-commutativity:
\begin{displaymath}
  {\letexpp{x}{e_1}{\letexp{y}{e_2}{e_0}}} ~\ctxeq~
  {\letexpp{y}{e_2}{\letexp{x}{e_1}{e_0}}}
\end{displaymath}
(provided $x$ does not occur free in $e_2$ and $y$ does not occur free in
$e_1$).  This equivalence, while valid for a pure language, is certainly not
valid for all effects (consider, for example, if there were an assignment
statement in $e_1$ or $e_2$). In other words, sampling and scoring are
commutative effects.
We conclude with two related work sections:
Section~\ref{sec:mismatch} demonstrates the correspondence between our
language model and others, notably that of \citet{Borgstrom:2016}, and
Section~\ref{sec:related-work} informally discusses other related work.

Throughout the paper, we limit proofs mostly to high-level sketches
and representative cases. Additional details and cases for some proofs
can be found in the appendices, along with a sketch of the steps of
the linear regression transformation using the equivalences we prove
in this paper.

\section{Probabilistic Language Model}
\label{sec:lang}

In this section we define our probabilistic language and its
semantics. The semantics consists of three parts:
\begin{itemize}
\item A notion of \emph{entropy} for modeling random behavior.
\item An \emph{evaluation} function that maps a program and entropy to
  a real-valued result and an importance weight. We define the
  evaluation function via an abstract machine. We then define a
  big-step semantics and prove it equivalent; the big-step formulation
  simplifies some proofs in Section~\ref{sec:commut} by making the
  structure of evaluation explicit.
\item A mapping to \emph{measures} over the real numbers, calculated
  by integrating the evaluation function with respect to the entropy
  space. A program with a finite, non-zero measure can be interpreted
  as an unnormalized probability distribution.
\end{itemize}

The structure of the semantics loosely corresponds to one
inference technique for probabilistic programming languages:
importance sampling. In an importance sampler, the entropy is
approximated by a pseudo-random number generator (PRNG); the
evaluation function is run many times with different initial PRNG
states to produce a collection of weighted samples; and the weighted
samples approximate the program's measure---either directly by
conversion to a discrete distribution of results, or indirectly via
computed statistical properties such as sample mean, variance, etc.

Our language is similar to that of \citet{Borgstrom:2016}, but with
the following differences:
\begin{itemize}
\item Our language requires \ttop{let}-binding of nontrivial
  intermediate expressions; this simplifies the semantics.  This
  restriction is similar to but looser than A-normal
  form~\cite{SabryFelleisen:1993}.
\item Our model of entropy is a finite measure space made of
  \emph{splittable} entropy points, rather than an infinite measure
  space containing \emph{sequences} of real numbers.
\item Our \ttop{sample} operation models a standard uniform random
  variable, rather than being parameterized over a distribution.
\end{itemize}
We revisit these differences in Section~\ref{sec:mismatch}.

\subsection{Syntax}
\label{sec:syntax}

\begin{figure}
  \centering
\begin{displaymath}
    \begin{altgrammar}
      v &::=& x \alt \lambdaexp{x}{e} \alt \const r
      &\hbox{Syntactic Values} \\
      e &::=& v \alt \appexp{v}{v} \alt
      \letexp{x}{e}{e}  & \hbox{Expressions} \\
      &\alt& \opexp{\op^n}{v_1, \dots, v_n} \alt
      \ifexp{v}{e}{e} \\
      &\alt& \sampleexp \alt \factorexp{v} \\
      \op^1 &::=&  \logop \alt \expop \alt \realop \alt \alt \dots
        & \hbox{Unary operations}\\
      \op^2 &::=& + \alt - \alt \times \alt \div \alt < \alt \leq \alt \dots
        & \hbox{Binary operations}\\
      \op^3 &::=& \ttop{normalinvcdf} \alt \ttop{normalpdf} \alt \dots
        & \hbox{Ternary operations}\\
      K &::=& \haltcont \alt {\letcont{x}{e}{K}}  &\hbox{Continuations}\\
    \end{altgrammar}
  \end{displaymath}
  \caption{Syntax of values, expressions, and continuations}
  \label{fig:syntax}
\end{figure}

The syntax of our language is given in Figure~\ref{fig:syntax}.
For simplicity, we require sequencing to be made explicit using
$\ttop{let}$.
There is a constant $\const{r}$ for each real number $r$, and there
are various useful primitive operations.

The \sampleexp form draws from a uniform distribution on $[0,1]$. Any
other standard real-valued distribution can be obtained by applying
the appropriate \emph{inverse cumulative distribution function}.
For example, sampling from a normal distribution can be expressed as follows:
\[
\opexptt{normal}{m, s} ~\triangleq~ \letexpp{u}{\sampleexp}{\opexptt{normalinvcdf}{u; m, s}}
\]
where $\opexptt{normalinvcdf}{u; m, s}$ is the least $x$ such that if
$X \sim \mathcal{N}(m,s^2)$ then $\mathrm{Pr}[X \leq x] = u$.

Finally, $\factorexp{v}$ weights (or ``scores'') the current
execution by the value $v$.

The language is untyped, but we express the scoping relations by rules
like typing rules.  We write $\hastype{\G}{e}{\Exptype}$ for the
assertion that $e$ is a well-formed expression whose free variables
are contained in the set \G, and similarly for values and
continuations.  The scoping rules are given in Figure~\ref{fig:scoping}.

\begin{figure}
\begin{mathpar}
\inferrule{x \in \G}{\hastype{\G}{x}{\Valtype}}

\inferrule{\hastype{\G,x}{e}{\Exptype}}
          {\hastype{\G}{\lambdaexp{x}{e}}{\Valtype}}

\hastype{\G}{\const{r}}{\Valtype}
\\\\
\inferrule{\hastype{\G}{v}{\Valtype}}
          {\hastype{\G}{v}{\Exptype}}

\inferrule{{\hastype{\G}{v_1}{\Valtype}} \\
           {\hastype{\G}{v_2}{\Valtype}}}
          {\hastype{\G}{\appexp{v_1}{v_2}}{\Exptype}}

\inferrule{{\hastype{\G}{e_1}{\Exptype}} \\
           {\hastype{\G,x}{e_2}{\Exptype}}}
          {\hastype{\G}{\letexp{x}{e_1}{e_2}}{\Exptype}}

\inferrule{\hastype{\G}{v_i}{\Valtype} \quad (\forall i \in \{1, \dots, n\} )}
          {\hastype{\G}{\opexp{\op^n}{v_1, \dots, v_n}}{\Exptype}}

\inferrule{\hastype{\G}{v}{\Valtype} \\
           \hastype{\G}{e_1}{\Exptype} \\
           \hastype{\G}{e_2}{\Exptype}}
          {\hastype{\G}{\ifexp{v}{e_1}{e_2}}{\Exptype}}
\\\\
\hastype{\G}{\sampleexp}{\Exptype}

\inferrule{\hastype{\G}{v}{\Valtype}}
          {\hastype{\G}{\factorexp v}{\Exptype}}
\\\\
\hastype{}{\haltcont}{\Conttype}

\inferrule{{\hastype{\setof{x}}{e}{\Exptype}}  \\
           {\hastype{}{K}{\Conttype}}}
          {\hastype{}{\letcont{x}{e}{K}}{\Conttype}}
\end{mathpar}
\caption{Scoping rules for values, expressions, and continuations}
\label{fig:scoping}
\end{figure}

\subsection{Modeling Entropy}
\label{sec:entropy}

The semantics uses an \emph{entropy} component as the source of
randomness.
We assume an entropy space \Ent along with its stock measure
$\mu_\Ent$.  We use \s and \t to range over values in \Ent.  When we
integrate over \s or \t, we implicitly use the stock measure; that is,
we write $\leb{f(\s)}{\s}$ to mean $\smint{f(\s)}{\mu_\Ent(d\s)}$.
Following \citet{CulpepperCobb:2017}, we assume that \Ent has the
following properties:

\begin{property}[Properties of Entropy]
  \label{lemma-entropy-def}
  \label{lemma-entropy-splitting}\
  \begin{enumerate}
  \item $\mu_\Ent(\Ent) = 1$

  \item There is a function $\pi_U : \Ent \to [0,1]$ such that for all
    measurable $f : [0,1] \to \Rplus$,
    \begin{displaymath}
      \leb {f(\pi_U(\s))} {\s} = \smintBB{0}{1}{f(x)}{\lambda(dx)}
    \end{displaymath}
    where $\lambda$ is the Lebesgue measure. That is, $\pi_U$
    represents a standard uniform random variable.

  \item There is a surjective pairing function
    $`\cons` : \H \times \H \to \H$, with projections $\pi_L$ and $\pi_R$, all measurable.

  \item The projections are measure-preserving: for all
    measurable $g : \Ent \times \Ent \to \Rplus$,
    \begin{displaymath}
      {\leb {g(\pi_L(\s), \pi_R(\s))} {\s}} = {\dint {g(\s_1,\s_2)} {\s_1} {\s_2}}
    \end{displaymath}
  \end{enumerate}

\end{property}

Since $\Ent \cong \Ent \times \Ent$ and thus $\Ent \cong \Ent^n ~(n
\geq 1)$, we can also use entropy to encode non-empty \emph{sequences}
of entropy values.

One model that satisfies these properties is the space of infinite
sequences of real numbers in $[0,1]$; $\pi_L$ and $\pi_R$ take the
odd- and even-indexed subsequences, respectively, and $\pi_U$ takes
the first element in the sequence.
Another model is the space of infinite sequences of bits, where
$\pi_L$ and $\pi_R$ take odd and even subsequences and $\pi_U$
interprets the entire sequence as the binary expansion of a number in
$[0,1]$.
%
%
It is tempting to envision entropy as infinite binary trees labeled
with real numbers in $[0,1]$, but the pairing function is not
surjective.

We also use Tonelli's Theorem:

\begin{lemma}[Tonelli]
  \label{lemma-tonelli}
  Let $f: \Ent \times \Ent \to \Rplus$ be measurable.  Then
  \begin{displaymath}
    {\leb {\left({\leb {f(\s_1,\s_2)} {\s_1}} \right)} {\s_2}}
    =
    {\leb {\left({\leb {f(\s_1,\s_2)} {\s_2}} \right)} {\s_1}}
  \end{displaymath}
\end{lemma}


\subsection{Operational Semantics}
\label{sec:op-sem}

\subsubsection{Small-Step Semantics}
\label{sec:op-sem-small}

We define evaluation via an abstract machine with a small-step
operational semantics.  The semantics rewrites configurations
\nconfig{\s}{e}{K}{\t}{w} consisting of:
\begin{itemize}
\item an entropy $\s$ (representing the ``current'' value of the
  entropy),
\item a closed expression $e$,
\item a closed continuation $K$,
\item an entropy $\t$ (encoding a stack of entropies, one for each frame of $K$), and
\item a positive real number $w$ (representing the weight of the
  current run)
\end{itemize}
The rules for the semantics are given in Figure~\ref{fig:op-sem}.

\begin{figure}
  \centering
 \begin{displaymath}
  \begin{array}{lcl}
    \nconfig{\s}{\letexp{x}{e_1}{e_2}}{K}{\t}{w}
    &\to&
    {\nconfig{\pi_L(\s)}{e_1}
          {\letcont{x}{e_2}{K}}
          {\pi_R(\s) \cons \t} {w}} \\

    {\nconfig{\s}{v}{\letcont{x}{e_2}{K}}{\s' \cons \t} {w} }
    &\to&
    {\nconfig{\s'}{e_2[v/x]}{K} {\t} {w}} \\

    {\nconfig {\s} {\appexp{\lambdaexpp{x}{e}}{v}}{K}{\t}{w}}
    &\to&
    {\nconfig {\s} {e[v/x]}{K} {\t}{w}}\\

    {\nconfig {\s}{\sampleexp}{K} {\t} {w}} &\to&
    {\nconfig {\pi_R(\s)} {\const{\pi_U(\pi_L(\s))}} {K} {\t} {w}} \\

    {\nconfig {\s}{\opexp{\op^n}{v_1, \dots, v_n}}{K} {\t}{w}} &\to&
    {\nconfig {\s}{\deltafun(\op^n, v_1, \dots, v_n)}{K} {\t}{w}} \hbox{ (if defined)}\\

    {\nconfig {\s}{\ifexp{\const r}{e_1}{e_2}}{K}{\t}{w} } &\to&
    {\nconfig {\s}{e_1}{K}{\t}{w} \hbox{ (if $r > 0$ ) }} \\

    {\nconfig {\s}{\ifexp{\const r}{e_1}{e_2}}{K}{\t}{w} } &\to&
    {\nconfig {\s}{e_2}{K}{\t}{w} \hbox{ (if $r \leq 0$ ) }} \\

    {\nconfig {\s} {\factorexp{\const{r}}} {K} {\t} {w}}&\to&
    {\nconfig {\s} {\const{r}} {K}{\t} {r \times w}} \hbox{ (provided $r > 0$)}
  \end{array}
\end{displaymath}

  \caption{Small-step operational semantics}
  \label{fig:op-sem}
\end{figure}

The semantics uses continuations for sequencing and substitutions for
procedure calls. Since \letexp{x}{e_1}{e_2} is the only sequencing
construct, there is only one continuation-builder.
The first rule recurs into the right-hand side of a \ttop{let}, using
the left half of the entropy as its entropy, and saving the right
half for use with $e_2$.
The second rule (``return'') substitutes the value of the expression
into the body of the \ttop{let} and restores the top saved entropy
value for use in the body. More precisely, we view the third component
as an encoded pair of an entropy value and an encoded entropy stack,
as mentioned in Section~\ref{sec:entropy}.\footnote{We defer the
  explanation of the initial entropy stack to
  Section~\ref{sec:measures}.}  The entropy stack $\t$ and
continuation $K$ are always updated simultaneously.
The return rule can be written using explicit projections as follows:
\[
    {\nconfig{\s}{v}{\letcont{x}{e_2}{K}}{\t} {w} }
    \quad\to\quad
    {\nconfig{\pi_L(\t)}{e_2[v/x]}{K} {\pi_R(\t)} {w}} \\
\]
Note that in the return rule the current entropy $\s$ is dead.
Except for the entropy and weight, these rules are standard for a
continuation-passing interpreter for the $\lambda$-calculus with
\ttop{let}.

The $\deltafun$ partial function interprets primitive operations.  We
assume that all the primitive operations are measurable partial
functions returning real values, and with the exception of \realop,
they are undefined if any of their arguments is a closure.
A conditional expression evaluates to its first branch if the
condition is a positive real constant, its second branch if
nonpositive; if the condition is a closure, evaluation is stuck.
Comparison operations and the $\sfop{real?}$ predicate return $1$ for
truth and $0$ for falsity.

The rule for \sampleexp uses $\pi_U$ to extract from the entropy a
real value in the interval $[0,1]$. The entropy is split first, to
make it clear that entropy is never reused, but the leftover entropy
is dead per the return rule.
The rule for $\factorexp{v}$ weights the current execution by $v$,
provided $v$ is a positive number; otherwise, evaluation is stuck.

When reduction of an initial configuration halts properly, there are
two relevant pieces of information in the final configuration: the
result value and the weight. Furthermore, we are only interested in
real-valued final results. We define \emph{evaluation} as taking an
extra parameter $A$, a measurable set of reals. Evaluation produces a
positive weight only if the result value is in the expected set.

\begin{align*}
  \eval{\s}{e}{K}{\t}{w}{A} &=
  \begin{cases}
    w' & \begin{array}[t]{l}
      \text{if $\nconfig{\s}{e}{K}{\t}{w} \to^* \nconfig{\s'}{r}{\haltcont}{\t'}{w'}$,} \\
      \text{where $r \in A$}
    \end{array} \\
    0  & \begin{array}[t]{l} \text{otherwise} \end{array}
  \end{cases}
\end{align*}


We will also need approximants to $\rmop{eval}$:

\begin{align*}
  \aeval{n}{\s}{e}{K}{\t}{w}{A} &=
  \begin{cases}
    w' & \begin{array}[t]{l}
      \text{if $\nconfig{\s}{e}{K}{\t}{w} \to^* \nconfig{\s'}{r}{\haltcont}{\t'}{w'}$} \\
      \text{in $n$ or fewer steps, where $r \in A$}
    \end{array} \\
    0  & \begin{array}[t]{l} \text{otherwise} \end{array}
  \end{cases}
\end{align*}

The following lemmas are clear from inspection of the small-step
semantics.

\begin{lemma}
  \label{lemma-opsem-trivial}
  If $ {\nconfig {\s}{e}{K}{\t}{w}} \to {\nconfig {\s'}{e'}{K'}{\t'}{w'}} $
  then
  \begin{enumerate}
  \item ${\aeval{p+1} {\s}{e}{K}{\t}{w}{A}} = {\aeval{p}{\s'}{e'}{K'}{\t'}{w'}{A}}$
  \item ${\eval {\s}{e}{K}{\t}{w}{A}} = {\eval {\s'}{e'}{K'}{\t'}{w'}{A}}$
  \end{enumerate}
\end{lemma}

\begin{lemma}[weights are Linear]
  \label{lemma:linear-weights}\
  \begin{enumerate}
  \item Weights can be factored out of reduction sequences. That is,
    \begin{displaymath}
      {\nconfig{\s}{e}{K}{\t}{1}} \to^* {\nconfig{\s'}{e'}{K'}{\t'}{w'}},
    \end{displaymath}
    if and only if for any $w > 0$
    \begin{displaymath}
      {\nconfig{\s}{e}{K}{\t}{w}} \to^*
      {\nconfig{\s'}{e'}{K'}{\t'}{w' \times w}}
    \end{displaymath}
  \item Weights can be factored out of evaluation. That is, for all $w > 0$,
    \begin{displaymath}
      \eval{\s}{e}{K}{\t}{w}{A} = w \times
      \eval{\s}{e}{K}{\t}{1}{A},
    \end{displaymath}
    and similarly for $\rmop{eval}^{(n)}$.
  \end{enumerate}
\end{lemma}

\subsubsection{Big-Step Semantics}
\label{sec:big-step}

We regard the small-step semantics as normative, and we use it for our
primary soundness and completeness results.  However, for program
transformations it is useful to have a big-step semantics as well.  In
this section, we define a big-step semantics and characterize its
relation to the small-step semantics.

The big-step semantics is given in Figure~\ref{fig:big-step}.  It has
judgments of the form $\bigstep{\s}{e}{v}{w}$, where $\s$ is a value
of the entropy, $e$ is a closed expression, $v$ is a closed value, and
$w$ is a weight (a positive real number).  Its intention is that when
$e$ is supplied with entropy $\s$, it returns $v$ with weight $w$,
consuming some portion (possibly empty) of the given entropy $\s$.
The rules are those of a straightforward call-by-value \l-calculus,
modified to keep track of the entropy and weight.

\begin{figure}
  \centering
  \begin{mathpar}
    {\bigstep{\s}{\lambdaexp{x}{e}}{\lambdaexp{x}{e}}{1}}

    {\bigstep{\s}{\const{r}}{\const{r}}{1}} \\

    {\inferrule{\bigstep \s {e[v/x]} {v'} {w}} {\bigstep {\s} {\appexp
          {\lambdaexpp{x}{e}} {v}} {v'} {w}} }

    {\inferrule {{\bigstep {\pi_L(\s)} {e_1} {v_1} {w_1}}
        \\
        {\bigstep {\pi_R(\s)} {e_2[v_1/x]} {v_2} {w_2}}} {\bigstep
        {\s} {\letexp{x}{e_1}{e_2}} {v_2} {w_2 \times w_1}}}

    {\inferrule {\deltafun(\op^n, v_1, \dots, v_n) = v} {\bigstep {\s}
        {\opexp{\op^n}{v_1, \dots, v_n}} v 1}}
    \\
    {\inferrule {{\bigstep {\s} {e_1} {v} {w}} \\ r > 0} {\bigstep
        {\s} {\ifexp{\const r}{\e1}{\e2}} {v} {w}}}

    {\inferrule {{\bigstep {\s} {e_2} v w} \\ r \leq 0} {\bigstep {\s}
        {\ifexp {\const r} {e_1} {e_2}} {v} {w}}}
    \\
    {\bigstep {\s} {\sampleexp} {\const {\pi_U(\pi_L(\s))}} 1}

    {\inferrule {r > 0} {\bigstep {\s} {\factorexp{\const{r}}}
        {\const{r}} r}}
  \end{mathpar}
  \caption{Big-step operational semantics}
  \label{fig:big-step}
\end{figure}

The translation from big-step to small-step semantics is
straightforward:

\begin{theorem}[Big-Step to Small-Step]
  \label{thm:bigstep-smallstep} If $\bigstep{\s}{e}{v}{w}$, then for
  any $K$ and $\t$, there exists a $\s'$ such that
  \begin{displaymath}
    {\nconfig{\s}{e}{K}{\t}{1}} \to^* {\nconfig{\s'}{v}{K}{\t}{w}}
  \end{displaymath}
  
\end{theorem}

\begin{proof}
  By induction on the definition of $\Downarrow$.  We will
  show selected cases.

  \begin{Case}{${\bigstep{\s}{\lambdaexp{x}{e}}{\lambdaexp{x}{e}}{1}}$} The required small-step reduction
    is empty.  Similarly for $\const r$.
  \end{Case}

  \begin{Case}{${\bigstep {\s} {\sampleexp} {\const
          {\pi_U(\pi_L(\s))}} 1}$} The required reduction is the
    single step reduction
    \begin{displaymath}
       {\nconfig {\s}{\sampleexp}{K} {\t} {1}} \to
       {\nconfig {\pi_R(\s)} {\const{\pi_U(\pi_L(\s))}} {K} {\t} {1}}
    \end{displaymath} 
    Similarly for ${\factorexp{\const{r}}}$ and the ${\op^n}$ rules.
  \end{Case}

  \begin{Case}{\appexp{\lambdaexpp{x}{\e{}}}{v}} The rule is 
    \begin{displaymath}
      {\inferrule{\bigstep \s {e[v/x]} {v'} {w}}
                 {\bigstep {\s} {\appexp {\lambdaexpp{x}{e}} {v}} {v'} {w}} }
    \end{displaymath}
    By inversion, we have ${\bigstep \s {e[v/x]} {v'} {w}}$.  So the
    reduction sequence is:
    \begin{displaymath}
      \begin{array}{ll@{\quad}l}
        {\nconfig {\s} {\appexp{\lambdaexpp{x}{\e{}}}{v}}{K}{\t}{1}}
        \\
        \to^{\hphantom{*}} {\nconfig {\s} {\e{}[v/x]}{K} {\t}{1}} \\
        \to^* {\nconfig {\s'} {v'} {K} {\t}{w}} &\hbox{by the induction hypothesis}
      \end{array}
    \end{displaymath}
    Similarly for the \texttt{if} rules.
  \end{Case}

  \begin{Case}{\letexp{x}{\e1}{\e2}}  The rule is 
    \begin{displaymath}
      {\inferrule
        {{\bigstep {\pi_L(\s)} {e_1} {v_1} {w_1}}
          \\
         {\bigstep {\pi_R(\s)} {e_2[v_1/x]} {v_2} {w_2}}}
        {\bigstep {\s} {\letexp{x}{e_1}{e_2}} {v_2} {w_2 \times w_1}}}
    \end{displaymath}

    By inversion, we have ${\bigstep {\pi_L(\s)} {e_1} {v_1} {w_1}}$
    and ${\bigstep {\pi_R(\s)} {e_2[v_1/x]} {v_2} {w_2}}$.  So the
    required reduction sequence is:
    \begin{displaymath}
      \begin{array}{l@{\quad}l}
        \nconfig{\s}{\letexp{x}{\e1}{\e2}}{K}{\t}{1} \\
        \to^{\hphantom{*}} {\nconfig{\pi_L(\s)}{\e1}
          {\letcont{x}{\e2}{K}}
          {\pi_R(\s) \cons \t} {w}} \\
        \to^*  
           {\nconfig{\s'}{v_1}
        {\letcont{x}{\e2}{K}}
        {\pi_R(\s) \cons \t} {w_1}} \\
        \to^{\hphantom{*}}
        {\nconfig {\pi_R(\s)} {\e2[v_1/x]} {K} {\t} {w_1}} \\
        \to^*
        {\nconfig {\s''} {v_2} {K} {\t} {w_2 \times w_1}}
      \end{array}
    \end{displaymath}
    where the third line follows from the induction hypothesis, and
    the last line follows from the other induction hypothesis and the
    linearity of weights (Lemma~\ref{lemma:linear-weights}).    
  \end{Case}
\end{proof}

Note that the weak quantifier (``there exists a $\s'$'') corresponds
to the fact that the entropy is dead in the return rule.

In order to prove a converse, we need some additional results about
the small-step semantics. 

\begin{definition}
  \label{defn:cont-extension}
  Define $\succeq$ to be the smallest relation defined by the
  following rules:
  \begin{displaymath}
    \begin{array}{c}
      {\inferrule [Rule 1:] {} {(K, \t) \succeq (K, \t)}} \qquad
      {\inferrule [Rule 2:]
          {(K', \t') \succeq (K, \t)}
          {(\letcont{x}{e}{K'}, \s \cons \t') \succeq (K, \t)}}
    \end{array}
  \end{displaymath}
\end{definition}

\begin{lemma}
  \label{lemma:continuations-are-continuous}
  Let \begin{displaymath}
    {\nconfig {\s_1}{e_1}{K_1}{\t_1}{w_1}} \to {\nconfig
      {\s_2}{e_2}{K_2}{\t_2}{w_2}} \to \dots
  \end{displaymath}
  be a reduction sequence in the
  operational semantics.  Then for each $i$ in the sequence either
  \begin{enumerate}
  \item[a.] there exists a smallest $j \le i$ such that $e_j$ is a value and
    $K_j = K_1$ and $\t_j = \t_1$, or
  \item[b.] $(K_i,\t_i) \succeq (K_1,\t_1)$
  \end{enumerate}
\end{lemma}
\proofinappendix

The next result is an interpolation theorem, which imposes structure
on reduction sequences: any terminating computation starting with an
expression $e$ begins by evaluating $e$ to a value $v$ and then
sending that value to the continuation $K$.

\begin{theorem}[Interpolation Theorem]
  \label{thm:interpolation}
  If
  \begin{displaymath}
    {\nconfig{\s}{e}{K}{\t}{w}} \to^* {\nconfig{\s''}{v''}{\haltcont}{\t''}{w''}}
  \end{displaymath}
  then there exists a smallest $n$ such that for some quantities $\s'$, $v$, and $w'$,
  \begin{displaymath}
    {\nconfig{\s}{e}{K}{\t}{w}} \to^n {\nconfig{\s'}{v}{K}{\t}{w'
        \times w}}\to^*
    {\nconfig{\s''}{v''}{\haltcont}{\t''}{w''}}
  \end{displaymath}
\end{theorem}

\begin{proof}
  If $K = \haltcont$, then the result is trivial.  Otherwise, apply the
  invariant of the preceding lemma, observing that $(\haltcont, \t')
  \not\succeq (K, \t)$ and that weights are multiplicative.
\end{proof}


Note that both Lemma~\ref{lemma:continuations-are-continuous} and
Theorem~\ref{thm:interpolation} would be false if our language
contained jumping control structures like \ttop{call/cc}.

Finally, we show that in the interpolation theorem, $\s'$, $v$,
and $w'$ are independent of $K$.


\begin{theorem}[Genericity Theorem]
  \label{thm-generic-conts}
  Let $w_1 > 0$ and let $n$ be the smallest integer such that for some quantities $\s'$, $v$, and $w'$,
    $${\nconfig{\s} {e}{K_1}{\t_1}{w_1}} \to^n
        {\nconfig{\s'}{v}{K_1}{\t_1}{w' \times w_1}}$$
then for any $K_2$, $\t_2$, and $w_2$,
$${\nconfig{\s} {e}{K_2}{\t_2}{w_2}} \to^n
        {\nconfig{\s'}{v}{K_2}{\t_2}{w' \times w_2}}$$
\end{theorem}

\begin{proof}
  Let $R$ be the smallest relation defined by the rules
  \begin{displaymath}
    {\inferrule {} {((K_1,\t_1), (K_2,\t_2)) \in R}}
    \qquad
    {\inferrule
      {((K,\t), (K', \t')) \in R}
      {((\letcont{x}{e}{K}, \s \cons \t),
        (\letcont{x}{e}{K'}, \s \cons \t')) \in R}}
  \end{displaymath}
  Extend $R$ to be a relation on configurations by requiring the
  weights to be related by a factor of $w_2 / w_1$ and the remaining
  components of the configurations to be equal.  It is easy to see, by
  inspection of the small-step rules, that $R$ is a bisimulation over
  the first $n$ steps of the given reduction sequence.
\end{proof}

We are now ready to state the converse of Theorem~\ref{thm:bigstep-smallstep}.

\begin{definition}
  \label{defn:halts}
  We say that a configuration $\nconfig{\s}{e}{K}{\t}{w}$ \emph{halts}
  iff
  \begin{displaymath}
    \nconfig{\s}{e}{K}{\t}{w} \to^* \nconfig{\s'}{v}{\haltcont}{\t'}{w'}
  \end{displaymath}
  for some $\s'$, $v$, $\t'$ and $w'$.  
\end{definition}

\begin{theorem}[Small-Step to Big-Step]
  \label{thm:small-step-big-step}
  If
  \begin{displaymath}
    {\nconfig{\s}{e}{K}{\t}{w}} \to^*
    {\nconfig{\s''}{v''}{\haltcont}{\t''}{w''}},
  \end{displaymath} then there exist $\s'$, $v'$ and $w'$ such
  that
  \begin{displaymath}
    \bigstep{\s}{e}{v'}{w'}
  \end{displaymath}
  and
  \begin{displaymath}
    {\nconfig{\s''}{v'}{K}{\t}{w' \times w}} \to^*
    {\nconfig{\s'}{v'}{\haltcont}{\t''}{w''}} 
  \end{displaymath}
\end{theorem}

\begin{proof}
  Given
  \begin{displaymath}
    {\nconfig{\s}{e}{K}{\t}{w}} \to^*
    {\nconfig{\s''}{v''}{\haltcont}{\t''}{w''}} 
  \end{displaymath}
  apply the Interpolation
  Theorem (Theorem~\ref{thm:interpolation}) to get $n$, $\s'$, $v$,
  and $w'$ 
  such that
  \begin{displaymath}
    {\nconfig{\s}{e}{K}{\t}{w}} \to^n {\nconfig{\s'}{v}{K}{\t}{w'
        \times w}}\to^*
    {\nconfig{\s''}{v''}{\haltcont}{\t''}{w''}}
  \end{displaymath}

  This gives us the second part of the conclusion.  To get the first
  part, we proceed by (course-of-values) induction on $n$, and then by
  cases on $e$.

  \begin{Case}{$\lambdaexp{x}{e}$}  For configurations of the form
    ${\nconfig{\s}{\lambdaexp{x}{e}}{K}{\t}{w}}$, the expression is
    already a value, so $n$  is 0.  So set $v = \lambdaexp{x}{e}$
    and $w' = 1$, and observe that
    ${\bigstep{\s}{\lambdaexp{x}{e}}{\lambdaexp{x}{e}}1}$, as desired.
    The case of constants $\const r$ is similar.    
  \end{Case}

  \begin{Case}{\sampleexp}
    We know
    \begin{displaymath}
      {\nconfig {\s}{\sampleexp}{K} {\t} {w}} \to
      {\nconfig {\pi_R(\s)} {\const{\pi_U(\pi_L(\s))}} {K} {\t} {w}} \\
    \end{displaymath}
    so the value length is 1, and we also have
    ${\bigstep {\s} {\sampleexp} {\const {\pi_U(\pi_L(\s))}} 1}$, as desired.
    The cases of $\ttop{factor}$ and of ${\op^n}$ are similar.
  \end{Case}

  \begin{Case}{$\appexp{\lambdaexpp{x}{e}}{v}$}  Assume that the
    value length of
    ${\nconfig {\s} {\appexp{\lambdaexpp{x}{\e{}}}{v}}{K}{\t}{w}}$ is $n+1$.
    So we have
    \begin{displaymath}
      {\nconfig {\s} {\appexp{\lambdaexpp{x}{\e{}}}{v}}{K}{\t}{w}}
      \to
      {\nconfig {\s} {\e{}[v/x]}{K} {\t}{w}}
      \to^n
      {\nconfig {\s'} {v'} {K} {\t}{w'\times w}}
    \end{displaymath}

    By induction, we have ${\bigstep{\s}{\e{}[v/x]}{v'}{w'}}$.  Hence,
    by the big-step rule for \l-expressions, we have
    ${\bigstep {\s}
      {\appexp {\lambdaexpp{x}{e}} {v}}
      {v'}
      {w'}}$, as desired.  The cases for conditionals are similar.    
  \end{Case}

  \begin{Case}{$\letexp{x}{\e1}{\e2}$}  Assume the value length of
    ${\nconfig{\s}{\letexp{x}{\e1}{\e2}}{K}{\t}{w}}$ is $n$.
    Then the first $n$ steps of its reduction sequence must be
    \begin{displaymath}
      \begin{array}{ll@{\quad}l}
        {\nconfig{\s}{\letexp{x}{\e1}{\e2}}{K}{\t}{w}} \\
        \tophantom
        {\nconfig
          {\pi_L(\s)}{\e1}
                 {\letcont{x}{\e2}{K}}
        {\pi_R(\s) \cons \t}
        {w}} \\
        \to^{m}
        {\nconfig {\s'} {v_1} {\letcont{x}{\e2}{K}} {\pi_R(\s) \cons
        \t} {w_1 \times w}} \\
        \tophantom
        {\nconfig {\pi_R(\s) \cons \t} {\e2[v_1/x]}
        {K} {\t} {w_1 \times w}} \\
        \to^{p}
        {\nconfig {\s''} {v} {K} {\t} {w_2 \times w_1 \times w}}        
      \end{array}
    \end{displaymath}
    where $m$ and $p$ are the value lengths of the
    configurations on the second and fourth lines, respectively.
    So $n = m + p + 2$, and we can apply the
    induction hypothesis to the two relevant configurations.  Applying
    the induction hypothesis twice, we get
    \begin{mathpar}
      {\bigstep {\pi_L(\s)} {e_1} {v_1} {w_1}}
      \and
      \text{and}
      \and
      {\bigstep{\pi_R(\s)} {e_2[v_1/x]} {v_2} {w_2}}\ .
    \end{mathpar}
    Hence, by the big-step rule for \texttt{let}, we conclude that
    \begin{displaymath}
      {\bigstep {\s} {\letexp{x}{e_1}{e_2}} v {w_2 \times w_1}}
    \end{displaymath}
    as desired.  
  \end{Case}
\end{proof}


\subsection{From Evaluations to Measures}
\label{sec:measures}

Up to now, we have considered only single runs of the machine, using
particular entropy values.  To obtain the overall meaning of the
program we need to integrate over all possible values of the entropies
$\s$ and $\t$:

\begin{definition}
  The \emph{measure} of $e$ and $K$ is the measure on the reals
  defined by
  \begin{displaymath}
    \Meas{e}{K}{A} = {\dint {\eval{\s}{e}{K}{\t}{1}{A}} {\s} {\t}}
  \end{displaymath}
  for each measurable set $A$ of the reals.
\end{definition}

This measure is similar to both Culpepper and Cobb's $\mu_e(A)$ and
Borgstr\"{o}m et al.'s $\sem{e}_{\mathrm{S}}(A)$, but whereas they
define measures on arbitrary syntactic values, our $\Meas{e}{K}{-}$ is
a measure on the reals.
Furthermore, whereas their measures represent the meanings of intermediate
expressions, our measure---due to the inclusion of the continuation
argument $K$---represents the meanings of whole programs.

The simplicity of the definition above relies on the mathematical trick
of encoding entropy stacks as entropy values; if we represented stacks
directly the number of integrals would depend on the stack depth.
Note that even for the base continuation ($K = \haltcont$) we still
integrate with respect to both $\s$ and $\t$.
Since $\Ent \not\cong \Ent^0$, there is no encoding for an empty stack
as an entropy value; we cannot just choose a single arbitrary
$\t_{\mathrm{init}}$ because $\mu_\Ent(\{\t_{\mathrm{init}}\}) = 0$.
But since evaluation respects the stack discipline, it produces the
correct result for any initial $\t_{\mathrm{init}}$. So we integrate
over \emph{all} choices of $\t_{\mathrm{init}}$, and since
$\mu_\Ent(\Ent) = 1$ the empty stack ``drops out'' of the integral.

As before, we will also need the approximants:
\begin{displaymath}
  \aMeas{n}{e}{K}{A} = {\dint {\aeval{n} {\s}{e}{K}{\t}{1}{A}} {\s} {\t}}
\end{displaymath}

For these integrals to be well-defined, of course, we need to know
that eval and its approximants are measurable.

\begin{lemma}[eval is measurable]
  \label{lemma-eval-measurable}
  For any $e$, $K$, $w \geq 0$, $A \in \SigmaReal$, and $n$,
  ${\eval{\s}{e}{K}{\t}{w}{A}}$ and ${\aeval{n}{\s}{e}{K}{\t}{w}{A}}$
  are measurable in $\s$ and $\t$.
\end{lemma}
\begin{proof}
  The proof is based on the proof from \citet{Borgstrom:2017}. See
  \theappendixref for more details.
\end{proof}

The next lemma establishes some properties of $\mu$ and the
approximants $\mu^{(n)}$. In particular, it shows that $\mu$ is the
limit of the approximants.

\begin{lemma}[measures are monotonic]
  \label{lemma:monotonic-measures}
  \label{lemma:misc-properties}
  In the following, $e$ and $K$ range over closed expressions and
  continuations, and let $A$ range over measurable sets of reals.
  \begin{enumerate}
  \item $\Meas{e}{K}{A} \ge 0$
  \item for any $m$,  $\aMeas{m}{e}{K}{A} \ge 0$
  \item if $m \le n$, then
    $\aMeas{m}{e}{K}{A} \le \aMeas{n}{e}{K}{A} \le \Meas{e}{K}{A}$
  \item $\Meas{e}{K}{A} = \sup_n\setof{\aMeas{n} {e}{K}{A}}$
  \end{enumerate}
\end{lemma}

Finally, the next lemma's equations characterize how the approximant
and limit measures, $\mu^{(n)}$ and $\mu$, behave under the reductions
of the small-step machine.  Almost all the calculations in
Section~\ref{sec:log-rel} depend only on these equations.

\begin{lemma}
  The following equations hold for approximant measures:
  \label{lemma:meas-eqns}
  \label{lemma-let} \label{lemma-return}
  \label{lemma-appexp} \label{lemma-opexp} \label{lemma-ifexp}
  \label{lemma-sample} \label{lemma-factor}
  \begin{align*}
    \aMeas{p+1} {\letexp{x}{e_1}{e_2}} {K} {A}
    &= \aMeas{p} {e_1} {\letcont{x}{e_2}{K}} {A}
    \\
    \aMeas{p+1} {v} {\letcont {x}{e}{K}} {A}
    &= \aMeas{p} {e[v/x]} {K} {A}
    \\
    \aMeas{p+1} {\appexp {\lambdaexp {x}{e}} {v}} {K} {A}
    &= \aMeas{p} {e[v/x]} {K} {A}
    \\
    \aMeas{p+1}{\opexp{\op^n}{v_1, \dots, v_n}}{K}{A}
    &= \aMeas{p}{\deltafun(\op^n, v_1, \dots, v_n)}{K}{A} \quad\text{if defined}
    \\
    \aMeas{p+1}{\ifexp{\const r}{e_1}{e_2}}{K}{A}
    &= \aMeas{p}{e_1}{K}{A} \quad\text{if $r > 0$}
    \\
    \aMeas{p+1}{\ifexp{\const r}{e_1}{e_2}}{K}{A}
    &= \aMeas{p}{e_2}{K}{A} \quad\text{if $r \leq 0$}
    \\
    \aMeas{p+1} {\sampleexp} {K} {A}
    &= \smintBB{0}{1}{\aMeas {p} {\const{r}} {K} {A}}{dr}
    \\
    \aMeas{p+1} {\factorexp {\const r}} {K} {A}
    &= r \times {\aMeas {p} {\const r} {K} {A}} \quad\text{if $r > 0$}
  \end{align*}
  In addition, the analogous index-free equations hold for the
  unapproximated (limit) measure $\Meas{-}{-}{-}$.
\end{lemma}

In general, the proofs of the equations of Lemma~\ref{lemma:meas-eqns}
involve unfolding the definition of the measure and applying
Lemma~\ref{lemma-opsem-trivial} under the integral. The proof for
\ttop{let} is representative:
\begin{proof}[Proof for \ttop{let}]
  \begin{align*}
    & {\aMeas {p+1} {\letexp{x}{e_1}{e_2}} {K}{A}} \\
    & = {\dint {\aeval {p+1} {\s}
                        {\letexp{x}{e_1}{e_2}}
                        {K} {\t} 1{A}}
           {\s} {\t}} \\
    & = {\dint {\aeval {p} {\pi_L(\s)} {e_1} {\letcont {x}{e_2}{K}}
                            {{\pi_R(\s)} \cons {\t}}
                            1{A}}
           {\s} {\t}}
    \tag{Lemma~\ref{lemma-opsem-trivial}} \\
    & = {\tint {\aeval {p} {\s'} {e_1} {\letcont {x}{e_2}{K}}
                            {{\s''} \cons {\t}}
                            1{A}}
           {\s'}  {\s''} {\t}}
    \tag{Property~\ref{lemma-entropy-splitting}.4 on $\s$} \\
    & = {\dint {\aeval {p} {\s'} {e_1} {\letcont {x}{e_2}{K}}
                            {\pi_L(\t') \cons \pi_R(\t')}
                            1{A}} {\s'} {\t'}}
    \tag{Property~\ref{lemma-entropy-splitting}.4 on $\t'$} \\
    & = {\dint {\aeval {p} {\s'} {e_1} {\letcont {x}{e_2}{K}}
                            {\t'}
                            1{A}}
           {\s'} {\t'}}
    \tag{${\pi_L(\t') \cons \pi_R(\t')} = \t'$} \\
    & = {\aMeas {p} {e_1} {\letcont {x}{e_2}{K}}{A}}
  \end{align*}
\end{proof}
The proof for $\ttop{factor}$ additionally uses linearity
(Lemma~\ref{lemma:linear-weights}), and the proof for $\sampleexp$
additionally uses Property~\ref{lemma-entropy-def}.2.

So far, our semantics speaks directly only about the meanings of whole
programs.
In the following sections, we
develop a collection of relations for expressions and
ultimately show that they respect the contextual ordering relation
on expression induced by the semantics of whole programs.


\section{The Logical Relation}
\label{sec:log-rel}

In this section, we define a step-indexed logical relation on values,
expressions, and continuations, and we prove the Fundamental Property
(a form of reflexivity) for our relation.

We begin by defining step-indexed logical relations on \emph{closed}
values, \emph{closed} expressions, and continuations (which are always
closed) as follows:
\begin{displaymath}
  \wider
  \begin{array}{l@{{}\iff{}}l}
    \ivalrel{}{n}{v_1}{v_2}
    & {\renewcommand\arraystretch{1.0}
        \block{
          v_1 = v_2 = \const r \hbox{ for  some } r \\
          \lor\ (\block {v_1 = {\lambdaexp{x}{e}} \land v_2 = {\lambdaexp{x}{e'}} \\
            \land\ (\forall m < n) (\forall v,v')[\valrel{}{m}{v}{v'} \implies
              \exprel{}{m}{e[v/x]}{e'[v'/x]}])}}}
    \\
    \exprel{}{n}{e}{e'}
    & \mathblock{
      (\forall m \le n)(\forall K,K')(\forall A \in \SigmaReal) \\
      \quad\protect{[\icontrel{}{m}{K}{K'} \implies \aMeas{m}{e}{K}{A} \le \Meas{e'}{K'}{A}]}
    }
    \\
    \contrel{}{n}{K}{K'}
    & \mathblock{
      (\forall m \le n)(\forall v,v')(\forall A \in \SigmaReal) \\
      \quad\protect{[\valrel{}{m}{v}{v'} \implies \aMeas{m}{v}{K}{A} \le \Meas{v'}{K'}{A}]}}
  \end{array}
\end{displaymath}
The definitions are well-founded because $\ivalrelset{}{-}$ refers to
$\iexprelset{}{-}$ at strictly smaller indexes.
Note that for all $n$, $\ivalrelset{}{n} \supseteq \ivalrelset{}{n+1}
\supseteq \ldots$, and similarly for $\iexprelset{}{}$ and
$\icontrelset{}{}$. That is, at higher indexes the relations make more
distinctions and thus relate fewer things.

We use $\g$ to range over substitutions of closed values for
variables, and we define $\ityenvrelset{}{n}$ by lifting
$\ivalrelset{}{n}$ to substitutions as follows:
\begin{displaymath}
  \ityenvrel{n}{\G}{\g}{\g'} \iff
  \begin{array}[t]{l}
    \dom(\g) = \dom(\g') = \G \\
    \land\ \forall x \in \G, {\valrel{}{n}{\g(x)}{\g'(x)}}
  \end{array}
\end{displaymath}

Last, we define the logical relations on open terms.  In each case,
the relation is on terms of the specified sort that are well-formed
with free variables in \G:
\begin{displaymath}
  \renewcommand\arraystretch{1.5}
  \begin{array}[t]{l@{{}\iff{}}l}
    \ivalrel{\G}{}{v}{v'}
    & (\forall n)(\forall \g, \g')[\ityenvrel{n}{\G}{\g}{\g'}
      \implies \ivalrel{}{n}{v\g}{v'\g'}] \\
    \iexprel{\G}{}{e}{e'}
    & (\forall n)(\forall \g, \g')[\ityenvrel{n}{\G}{\g}{\g'}
      \implies \exprel{}{n}{e\g}{e'\g'}] \\
    \icontrel{}{}{K}{K'}
    & (\forall n) \icontrel{}{n}{K}{K'}
  \end{array}
\end{displaymath}
The limit relation $\contrelset{}{}$ is not indexed by $\G$ because we
work only with closed continuations.

Our first goal is to show the so-called fundamental property of logical
relations:

$$\hastype{\G}{e}{\Exptype} \implies \exprel{\G}{}{e}{e}$$


We begin with a series of compatibility lemmas.  These show that the
logical relations form a congruence under (``are compatible with'')
the scoping rules of values, expressions, and continuations.
Note the correspondence between the scoping rules of
Figure~\ref{fig:scoping} and the compatibility rules of
Figure~\ref{fig:compatibility}.

\begin{figure}
  \begin{mathpar}
    \inferrule{x \in \G}
              {\ivalrel{\G}{}{x}{x}}

    \inferrule{\iexprel{\G,x}{}{e}{e'}}
              {\ivalrel{\G}{}{\lambdaexp{x}{e}}{\lambdaexp{x}{e'}}}

    \inferrule{}{\ivalrel{\G}{}{\const r}{\const r}}

    \inferrule{\ivalrel{\G}{}{v}{v'}}
              {\iexprel{\G}{}{v}{v'}}

    \inferrule{\ivalrel{\G}{}{v_1}{v_1'} \\
               \ivalrel{\G}{}{v_2}{v_2'}}
              {\iexprel{\G}{}{\appexp{v_1}{v_2}}{\appexp{v_1'}{v_2'}}}

    \inferrule{\exprel{\G}{}{e_1}{e_1'} \\
               \exprel{\G,x}{}{e_2}{e_2}}
              {\exprel{\G}{}{\letexp{x}{e_1}{e_2}}{\letexp{x}{e_1'}{e_2'}}}

    \inferrule{\ivalrel{\G}{}{v_i}{v_i'} \quad(i \in \{1, \dots, k\})}
              {\iexprel{\G}{}{\opexp{\op^k}{v_1, \dots, v_k}}{\opexp{\op^k}{v_1', \dots, v_k'}}}

    \inferrule{\valrel{\G}{}{v}{v'} \\
               \exprel{\G}{}{e_1}{e_1'} \\
               \exprel{\G}{}{e_2}{e_2}}
              {\exprel{\G}{}{\ifexp{v}{e_1}{e_2}}{\ifexp{v'}{e_1'}{e_2'}}}
    \\
    \inferrule{}{\iexprel{}{}{\sampleexp}{\sampleexp}}

    \inferrule{\ivalrel{\G}{}{v}{v'}}
              {\iexprel{\G}{}{\factorexp{v}} {\factorexp{v'}}}
    \\
    \inferrule{}{\contrel{}{}{\haltcont}{\haltcont}}

    \inferrule
        {\exprel{\{x\}}{}{e_1}{e_2} \\
         \icontrel{}{}{K}{K'}}
        {\icontrel{}{}{\letcont{x}{e}{K}}{\letcont{x}{e'}{K'}}}
  \end{mathpar}
  \caption{Compatibility rules for the logical relation} \label{fig:compatibility}
\end{figure}

\begin{lemma}[Compatibility] \label{lemma:compatibility}
  The implications summarized as inference rules in
  Figure~\ref{fig:compatibility} hold.
\end{lemma}

Most parts of the lemma follow by general reasoning about the
$\lambda$-calculus, the definitions of the logical relations, and
calculations involving $\mu^{(n)}$ and $\mu$ using
Lemma~\ref{lemma:meas-eqns}. The proof for application is
representative:

\begin{proof}[Proof for app]
  We must show that if ${\ivalrel{\G}{}{v_1}{v_1'}}$ and ${\ivalrel{\G}{}{v_2}{v_2'}}$, then
  ${\iexprel{\G}{}{\appexp{v_1}{v_2}}{\appexp{v_1'}{v_2'}}}$.

  Choose $n$, and assume $\ityenvrel{n}{\G}{\g}{\g'}$.  Then
  ${\ivalrel{}{n}{v_1\g\,}{v_1'\g'}}$ and $\ivalrel{}{n}{v_2\g}{v_2'\g'}$.
  We must show
  ${\iexprel{}{n}{\appexp{v_1\g\,}{v_2\g}}{\appexp{v_1'\g'}{v_2'\g'}}}$.

  If $v_1\g$ is of the form $\const r$, then $\aMeas{m}{v_1\g}{K}{A} =
  0$ for any $m$, $K$, and $A$, so the conclusion holds by
  Lemma~\ref{lemma:misc-properties}.

  Otherwise, assume $v_1\g$ is of the form \lambdaexp{x}{e}, and so
  $v_1'\g'$ is of the form \lambdaexp{x}{e'}.
  So choose $m \le n$ and $A$, and let  $\contrel{}{m}{K}{K'}$.
  We must show that
  \[\aMeas {m} {\appexp{\lambdaexp{x}{e\g}}{v_2\g}} {K}{A}
  \le \Meas {\appexp{\lambdaexp{x}{e'\g'}}{v_2'\g'}}{K'}{A}.\]

  If $m = 0$ the left-hand side is 0 and the inequality holds trivially.
  So consider $m \ge 1$.  Since all
  the relevant terms are closed and the relations on closed terms are
  antimonotonic in the index, we have
  ${\ivalrel{}{m}{\lambdaexp{x}{e\g}}{\lambdaexp{x}{e'\g'}}}$ and
  ${\ivalrel{}{m-1}{v_1\g\,}{v_1'\g'}} $.  Therefore
  ${\iexprel{}{m-1}{e\g[v_2\g/x]}{e'\g'[v_2'\g'/x]}}$.

  Now, $\configsw{\appexp{\lambdaexp{x}{e\g}\,}{v_2\g}}{K} \to
  \configsw{e\g[v_2\g/x]}{K}$, and similarly for the primed side.  So we
  have
  \begin{align*}
    \aMeas {m} {\appexp{\lambdaexp{x}{e\g}\,}{v_2\g}}{K}{A}
    &= \aMeas {m-1} {e\g[v_2\g/x]}{K}{A}  \tag{Lemma~\ref{lemma-appexp}}
    \\ &\le \Meas {e'\g'[v_2'\g'/x]}{K'}{A}  \tag{by
      ${\iexprel{}{m-1}{e\g[v_2\g/x]}{e'\g'[v_2'\g'/x]}}$}
    \\ &= \Meas {\appexp{\lambdaexp{x}{e'\g'}\,}{v_2'\g'}}{K'}{A}
  \end{align*}
\end{proof}

More detailed proofs can be found in \theappendixref.


Now we can prove the Fundamental Property:

\begin{theorem}[Fundamental Property]\
  \begin{enumerate}
  \item $\hastype{\G}{e}{\Exptype}  \implies \iexprel{\G}{}{e}{e}$
  \item $\hastype{\G}{v}{\Valtype}  \implies \ivalrel{\G}{}{v}{v}$
  \item $\hastype{}{K}{\Conttype} \implies \forall n, \icontrel{}{n}{K}{K}$
  \end{enumerate}
\end{theorem}

\begin{proof}
  By induction on the derivation of $\hastype{\G}{e}{\Exptype}$, etc,
  applying the corresponding compatibility rule from
  Lemma~\ref{lemma:compatibility} at each point.
\end{proof}

The essential properties of the logical relation we wish to hold are
soundness and completeness with respect to the contextual ordering.
We address these properties in Section~\ref{sec:ciu-ctx} after taking
a detour to define another useful intermediate relation,
$\iciurelset{\G}{}$, and establish its equivalence to
$\iexprelset{\G}{}$.


\section{CIU Ordering}
\label{sec:ciu}

The CIU (``closed instantiation of uses'') ordering of two terms
asserts that they yield related observable behavior under a single
substitution and a single continuation.  We take ``observable
behavior'' to be a program's measure over the reals, as we did for the
logical relations.

\begin{definition}
  \label{def-ciu} \
  \begin{enumerate}
  \item If $e$ and $e'$ are closed expressions, then
    $\iciurel{}{}{e}{e'}$ iff for all closed $K$ and measurable $A$,
    $\Meas{e}{K}{A} \le \Meas{e'}{K}{A}$.
  \item If $\hastype{\G}{e}{\Exptype}$ and $\hastype{\G}{e'}{\Exptype}$,
    then $\iciurel{\G}{}{e}{e'}$ iff for
    all closing substitutions \g, $\iciurel{}{}{e\g}{e'\g}$.
  \end{enumerate}
\end{definition}

Since it requires considering only a single substitution and a single
continuation rather than related pairs, it is often easier to prove
particular expressions related by $\iciurelset{\G}{}$. But in fact,
this relation coincides with the logical relation, as we demonstrate
now. One direction is an easy consequence of the Fundamental Property.

\begin{lemma}[$\iexprelset{}{} \subseteq \iciurelset{}{}$]
  \label{lemma-e-subset-ciu}
  If $\iexprel{\G}{}{e}{e'}$ then $\iciurel{\G}{}{e}{e'}$.
\end{lemma}
\begin{proof}
  Choose a closing substitution \g, a closed continuation $K$, and $A \in \SigmaR$.  By
  the Fundamental Property, we have for all $n$, $\ityenvrel{n}{\G}{\g}{\g}$
   and $\icontrel{}{n}{K}{K}$.  Therefore,
  for all $n$, $\aMeas {n}{e\g}{K}{A} \le \Meas{e'\g}{K}{A}$.  So
  \begin{displaymath}
    \Meas {e\g}{K}{A} = \sup_n\setof{\aMeas {n}{e\g}{K}{A}}
    \le \Meas{e'\g}{K}{A} .
  \end{displaymath}
\end{proof}

In the other direction:

\begin{lemma}[$\iexprelset{\G}{} \circ \iciurelset{\G}{} \subseteq \iexprelset{\G}{}$]
  \label{lemma-exprel-absorbtive}
  If $\iexprel{\G}{}{e_1}{e_2}$ and $\iciurel{\G}{}{e_2}{e_3}$,
  then $\iexprel{\G}{}{e_1}{e_3}$.
\end{lemma}
\begin{proof}
  Choose $n$ and ${\ityenvrel{\G}{n}{\g}{\g'}}$.  We must show that
  $\iexprel{\G}{n}{e_1\g}{e_3\g'}$.  So choose $m \le n$,
  ${\icontrelalt{}{m}{K}{K'}}$, and $A \in \SigmaR$.  Now we must show 
  $\aMeas{m}{e_1\g}{K}{A} \le \Meas {e_3\g'}{K'}{A}$.

  We have  $\iexprel{\G}{}{e_1}{e_2}$
  and ${\ityenvrel{\G}{n}{\g}{\g'}}$,
  so $\iexprel{}{n}{e_1\g}{e_2\g'}$, and by  $m \le n$ we have
  $\iexprel{}{m}{e_1\g}{e_2\g'}$. So
  \begin{align*}
    {\aMeas {n} {e_1\g} {K}{A}}
    & \le {\Meas {e_2\g'} {K'}{A}}
    \tag{by $\iexprel{}{m}{e_1\g}{e_2\g'}$} \\
    & \le {\Meas {e_3\g'} {K'}{A}}
    \tag{by $\iciurel{}{}{e_2}{e_3}$} \\
  \end{align*}
  Therefore $\iexprel{\G}{}{e_1}{e_3}$.
\end{proof}

\begin{lemma}[$\iciurelset{}{} \subseteq \iexprelset{}{}$]
  \label{lemma-ciu-subset-e}
  If $\iciurel{\G}{}{e}{e'}$ then $\iexprel{\G}{}{e}{e'}$.
\end{lemma}

\begin{proof}
  Assume $\iciurel{\G}{}{e}{e'}$. By the Fundamental Property, we know
  $\iexprel{\G}{}{e}{e}$.  So we
  have $\iexprel{\G}{}{e}{e}$ and $\iciurel{\G}{}{e}{e'}$.  Hence,
  by Lemma~\ref{lemma-exprel-absorbtive},  $\iexprel{\G}{}{e}{e'}$.
\end{proof}

\begin{theorem}
  \label{thm-ciu-equals-e}
  $\iciurel{\G}{}{e}{e'}$ iff $\iexprel{\G}{}{e}{e'}$.
\end{theorem}

\begin{proof}
  Immediate from Lemmas~\ref{lemma-e-subset-ciu} and
  \ref{lemma-ciu-subset-e}.
\end{proof}

\section{Contextual Ordering}
\label{sec:ciu-ctx}

Finally, we arrive at the contextual order relation. We define the
contextual ordering as the largest preorder that is both
\emph{adequate}---that is, it distinguishes terms that have different
observable behavior by themselves---and \emph{compatible}---that is,
closed under context formation, and we show that the contextual
ordering, the CIU ordering, and the logical relation all coincide.
Thus in order to show two terms contextually ordered, it suffices to
use the friendlier machinery of the CIU relation.

\begin{definition}[{$\ictxrelset{\G}{}$}]
  \label{defn-ctxrel}
  ${\ictxrelset{}{}}$ is the largest family of relations {$\iarbrelset{\G}{}$}
  such that:

  \begin{enumerate}
  \item {$\iarbrelset{}{}$} is adequate, that is, if $\G = \emptyset$, then
    $\iarbrel{\G}{}{e}{e'}$ implies that for all measurable
    subsets $A$ of the reals, $\Meas{e}{\haltcont}{A} \le \Meas{e'}{\haltcont}{A}$.
  \item For each \G, ${\iarbrelset{\G}{}}$ is a preorder.
  \item The family of relations $\iarbrelset{}{}$ is compatible, that
    is, it is closed under the type rules for expressions:
    \begin{enumerate}

    \item If ${\iarbrel{\G,x}{}{e}{e'}}$, then
      ${\iarbrel {\G}{} {\lambdaexp{x}{e}} {\lambdaexp{x}{e'}}}$.

    \item If $\iarbrel{\G}{}{v_1}{v_1'}$ and
      $\iarbrel{\G}{}{v_2}{v_2'}$, then
      $\iarbrel{\G}{}{\appexp{v_1}{v_2}}{\appexp{v_1'}{v_2'}}$.

    \item If ${\iarbrel{\G}{}{v}{v'}}$, then
      ${\iarbrel{\G}{}{\factorexp{v}}{\factorexp{v'}}}$.

    \item If ${\iarbrel{\G}{}{e_1}{e_1'}}$ and
      ${\iarbrel{\G, x}{}{e_2}{e_2'}}$,\\ then
      ${\iarbrel{\G}{}
        {\letexp{x}{e_1}{e_2}}
        {\letexp{x}{e_1'}{e_2'}}}$.

    \item If $\iarbrel{\G}{}{v_1}{v_1'}$, \dots,
      $\iarbrel{\G}{}{v_n}{v_n'}$,\\ then
       ${\iarbrel{\G}{}{\opexp{\op^n}{v_1, \dots, v_n}}
         {\opexp{\op^n}{v_1', \dots, v_n'}}}$.

     \item If  $\iarbrel{\G}{}{v}{v'}$,
       ${\iarbrel{\G}{}{e_1}{e_1'}}$, and
       ${\iarbrel{\G}{}{e_2}{e_2'}}$, \\ then
        ${\iarbrel{\G}{}{\ifexp{v}{e_1}{e_2}} {\ifexp{v'}{e_1'}{e_2'}}}$.

    \end{enumerate}
  \end{enumerate}
\end{definition}

Note, as usual, that the union of any family of relations satisfying these
conditions also satisfies these conditions, so the union of all of
them is the largest such family of relations.


We prove that $\exprelset{\G}{}$, $\ciurelset{\G}{}$, and
$\ctxrelset{\G}{}$ by first showing that
 $\exprelset{\G}{} \subseteq \ctxrelset{\G}{}$
and then that
 $\ctxrelset{\G}{} \subseteq \ciurelset{\G}{}$.
Then, having caught $\ctxrelset{\G}{}$ between $\exprelset{\G}{}$
and $\ciurelset{\G}{}$---two relations that we have already proven
equivalent---we conclude that all of the relations coincide.


First, we must show that $\exprelset{\G}{} \subseteq
\ctxrelset{\G}{}$. The heart of that proof is showing that
$\exprelset{\G}{}$ is compatible in the sense of
Definition~\ref{defn-ctxrel}. That is \emph{nearly} handled by the
existing compatibility rules for $\exprelset{\G}{}$
(Lemma~\ref{lemma:compatibility}), except for an occasional mismatch
between expressions and values---that is, between $\exprelset{\G}{}$
and $\valrelset{\G}{}$ in the rules. So we need a lemma to address
the mismatch (Lemma~\ref{lemma-e-implies-v-on-closed-values}), which
itself needs the following lemma due to \citet{Pitts:2010}.

\begin{lemma}
  \label{lemma-swap}
  If ${\icontrel{}{n}{K}{K'}}$ and ${\ivalrel{}n{v}{v'}}$, then
    \begin{displaymath}
      \icontrel{}
        {n+2}
        {\letcont{z}{\appexp{z}{v}}{K}}
        {\letcont{z}{\appexp{z}{v'}}{K'}}
    \end{displaymath}
\end{lemma}
\proofinappendix

\begin{lemma}
  \label{lemma-e-implies-v-on-closed-values}
  For all closed values $v$, if $\iexprel{}{}{v}{v'}$, then
  $\ivalrel{}{}{v}{v'}$.
\end{lemma}

\begin{proof}
  We will show that for all closed values $v$, $v'$, if $\iexprel{}{n+3}{v}{v'}$, then
  $\ivalrel{}{n}{v}{v'}$, from which the lemma follows.

  If $v = \const r$ and $v' = \const {r'}$, then $r = r'$ and thus
  $\ivalrel{}{}{\const r}{\const {r'}}$ because otherwise we would have
  $\Meas{\const r}{\haltcont}{\{r\}} = I_{\{r\}}(r) = 1$ and
  $\Meas{\const {r'}}{\haltcont}{\{r\}} = I_{\{r\}}(r') = 0$,
  violating the assumption $\iexprel{}{}{\const r}{\const {r'}}$.
  
  If only one of $v$ and $v'$ is a constant, then $\iexprel{}{n+3}{v}{v'}$
  is impossible, since constants and lambda-expressions are distinguishable
  by \realop (which requires 3 steps to do so).

  So assume $v = \lambdaexp{x}{e}$ and $v' =
  \lambdaexp{x}{e'}$. To establish $\ivalrel{}{n}{v}{v'}$, choose $m <
  n$ and $\ivalrel{}{m}{u}{u'}$.  We must show that
  $\iexprel{}{m}{e[u/x]}{e'[u'/x]}$.  To do that, choose $p \le m$,
  $\icontrel{}{p}{K}{K'}$, and $A \in \SigmaR$.  We must show that
  \begin{displaymath}
    {\aMeas p {e[u/x]} {K}{A}} \le {\Meas {e'[u'/x]} {K'}{A}}
  \end{displaymath}

  Let $K_1 = {\letcont {f} {\appexp{f}{u}}{K}}$ and $K_1' =
  {\letcont{f}{\appexp{f}{u'}}{K'}}$.  By monotonicity,
    $\ivalrel{}{p}{u}{u'}$.  By Lemma~\ref{lemma-swap},
    $\icontrel{}{p+2}{K_1'}{K_1'}$.
    Furthermore, $p \le m < n$, so $p+2 \le n+1$ and therefore
  ${\iexprel{}{p+2}{\lambdaexp{x}{e}} {\lambdaexp{x}{e'}}}$.  And
      furthermore, we have
      \begin{displaymath}
        \configsw{\lambdaexp{x}{e}}{K_1} \to
        {\configsw {\appexp {\lambdaexp{x}{e}} {u}} {K}}
        \to {\configsw {e[u/x]} {K}}
      \end{displaymath}
      and similarly on the primed side.

      We can put the results together to get
      \begin{align*}
        {\aMeas p {e[u/x]} K{A}}
        & = {\aMeas {p+2} {\lambdaexp{x}{e}} {K_1}{A}} \\
        & \le {\Meas {\lambdaexp{x}{e'}} {K_1'}{A}} \\
        & = {\Meas {e'[u'/x]} {K'}{A}}
      \end{align*}
\end{proof}

\begin{theorem}
  \label{thm-exp-subset-ctx}
  $\exprelset{\G}{} \subseteq \ctxrelset{\G}{}$.
\end{theorem}

\begin{proof}
  We will show that $\exprelset{}{}$ forms a family of reflexive
  preorders that is adequate and compatible.
  Each $\exprelset{\G}{}$ is reflexive by the Fundamental Property,
  and is a preorder because it is equal to $\ciurelset{\G}{}$, which
  is a preorder.
  To show that it is adequate, observe that
  $\contrel{}{}{\haltcont}{\haltcont}$ by Lemma~\ref{lemma:compatibility},
  hence for any measurable subset $A$ of reals,
  ${\exprel{\G}{}{e}{e'}}$ implies $\Meas{e}{\haltcont}{A} = \Meas{e'}{\haltcont}{A}$.

  The $\exprelset{}{}$-compatibility rules
  (Lemma~\ref{lemma:compatibility}) are almost exactly what is needed
  for $\ctxrelset{}{}$-compatibility.  The exceptions are in the
  application, operation, \ttop{if}, and \ttop{factor} rules
  where their hypotheses refer to $\valrelset{\G}{}$ rather than
  $\exprelset{\G}{}$.  We fill the gap with
  Lemma~\ref{lemma-e-implies-v-on-closed-values}.  We show how this is
  done for $\factorexp{v}$; the other cases are similar.

    \begin{align*}
      \iexprel{\G}{}{v}{v'} 
      & \implies \iciurel{\G}{}{v}{v'} \\
      & \implies (\forall \g)(\iciurel{\emptyset}{}{v\g}{v'\g}) \\
      & \implies (\forall \g)(\iexprel{\emptyset}{}{v\g}{v'\g}) \\
      & \implies (\forall \g)(\ivalrel{\emptyset}{}{v\g}{v'\g})
      \tag{Lemma~\ref{lemma-e-implies-v-on-closed-values}} \\
      & \implies (\forall \g)
                    (\iexprel{\emptyset}{}
                      {\factorexp {v\g}}
                      {\factorexp {v'\g}})
      \tag{Lemma~\ref{lemma:compatibility}} \\
      & \implies (\forall \g)
                    (\iciurel{\emptyset}{}
                      {\factorexp {v\g}}
                      {\factorexp {v'\g}}) \\
      & \implies \iciurel{\G}{}{\factorexp v}{\factorexp {v'}} \\
      & \implies \iexprel{\G}{}{\factorexp v}{\factorexp {v'}}
    \end{align*}
\end{proof}


Next, we must show that $\ctxrelset{\G}{} \subseteq \ciurelset{\G}{}$ by
induction on the closing substitution and then induction on the
continuation. We use the following two lemmas to handle the closing substitution.

\begin{lemma}
  \label{lemma-beta-value-ciu}
  If {} ${\hastype{\G,x}{e}{\Exptype}}$ and $\hastype{\G}{v}{\Exptype}$, then
  \[
    \iciurel{\G}{}{e[v/x]}{\appexp{\lambdaexp{x}{e}}{v}}
    \text{ ~and~ }
    \iciurel{\G}{}{\appexp{\lambdaexp{x}{e}}{v}}{e[v/x]}.
  \]
\end{lemma}

\begin{proof}
  Let \g be a closing substitution for \G. Then for any \s, closed
  $K$, and $w$, by Lemmas~\ref{lemma-appexp} and
  \ref{lemma:monotonic-measures}.4 we have
    \begin{displaymath}
      {\nconfig {\s} {\appexp{\lambdaexp{x}{e\g}}{v\g}} {K} {\t} {w}}
      \to
      {\configsw {e\g[v\g/x]} {K} }
    \end{displaymath}
    Therefore for any $A \in \SigmaR$,
    ${\Meas {\appexp{\lambdaexp{x}{e\g}}{v\g}} {K}{A}} = {\Meas {e\g[v\g/x]} {K}{A}}$.
\end{proof}

\begin{lemma}
  \label{lemma-ctx-closed-under-subst}
  If ${\ctxrel{\G,x}{}{e}{e'}}$, and ${\ctxrel {\G} {} {v} {v'}}$,
  then ${\ctxrel {\G} {} {e[v/x]} {e'[v'/x]}}$.
\end{lemma}

\begin{proof}
  From the assumptions and the compatibility of $\ctxrelset{}{}$, we
  have
    \begin{equation}
      \label{eq:a}
      {\ctxrel{\G}{} {\appexp {\lambdaexp{x}{e}} {v}} {\appexp {\lambdaexp{x}{e'}} {v'}}}
    \end{equation}

So now we have:

\begin{align*}
  & {\ciurel{\G}{} {e[v/x]} {\appexp {\lambdaexp{x}{e}} {v}}}
  \tag{Lemma~\ref{lemma-beta-value-ciu}} \\
  & \implies {\ctxrel{\G}{} {e[v/x]} {\appexp {\lambdaexp{x}{e}} {v}}}
  \tag{$\ciurelset{\G}{} \subseteq \ctxrelset{\G}{}$} \\
  & \implies {\ctxrel{\G}{} {e[v/x]} {\appexp {\lambdaexp{x}{e'}} {v'}}}
  \tag{Equation (\ref{eq:a}) and transitivity of ${\ctxrelset{\G}{}}$} \\
  & \implies \ctxrel{\G}{}{e[v/x]}{e'[v'/x]}
  \tag{Lemma~\ref{lemma-beta-value-ciu} and transitivity of ${\ctxrelset{\G}{}}$}
\end{align*}
\end{proof}


Now we are ready to complete the theorem.  Here we need to use
${\ciurelset{}{}}$ rather than ${\exprelset{}{}}$, so that we can deal
with only one continuation rather than two.

\begin{theorem}[$\ctxrelset{\G}{} \subseteq \ciurelset{\G}{}$]
  \label{thm-ctx-subset-ciu}
  If $\ctxrel{\G}{}{e}{e'}$, then $\ciurel{\G}{}{e}{e'}$
\end{theorem}

\begin{proof}
   By the preceding lemma, we have $\ctxrel{}{}{e\g}{e'\g}$.   So it
  suffices to show that for all $A \in \SigmaR$,
  if $\ctxrel{\emptyset}{}{e}{e'}$ and
  $\hastype{}{K}{\Conttype}$,
  then ${\Meas{e}{K}{A}} = {\Meas{e'}{K}{A}}$.

  The proof proceeds by induction on $K$ such that
  $\hastype{}{K}{\Conttype}$.  The
  induction hypothesis on
  $K$ is: for all closed $e$,
  $e'$, if ${\ctxrel{\emptyset}{}{e}{e'}}$, then
  ${\Meas e K A = \Meas {e'} K A}$.

  If $K = \haltcont$ and ${\ctxrel{\emptyset}{}{e}{e'}}$,
  then ${\Meas{e}{\haltcont}{A} = \Meas{e'}{\haltcont}{A}}$
  by the adequacy of $\ctxrelset{\emptyset}{}$.

  For the induction step, consider $\letcont{x}{e_1}{K}$, where
  {\hastype{x}{e_1}{\Exptype}}. Choose ${\ctxrel{\emptyset}{}{e}{e'}}$.
  We must show $\Meas{e}{\letcont{x}{e_1}{K}}{A} \le \Meas{e'}{\letcont{x}{e_1}{K}}{A}$.

  By the compatibility of $\ctxrelset{}{}$, we have
  \begin{equation}
    \label{eq:c}
    {\ctxrel {\emptyset} {} {\letexp{x}{e}{e_1}}{\letexp{x}{e'}{e_1}}}
  \end{equation}

  Then we have
  \begin{align*}
    \Meas{e}{\letcont{x}{e_1}{K}}{A}
    & = \Meas{\letexp{x}{e}{e_1}}{K}{A}
    \tag{Lemma~\ref{lemma-let}} \\
    & \le \Meas{\letexp{x}{e'}{e_1}}{K}{A}
    \tag{by IH at $K$, applied to (\ref{eq:c})} \\
    & = \Meas{e'}{\letcont{x}{e_1}{K}}{A}
    \tag{Lemma~\ref{lemma-let}}
  \end{align*}
Thus completing the induction step.
\end{proof}

Summarizing the results:

\begin{theorem}
  \label{thm-ciu=ctx}
  For all \G, $\ciurelset{\G}{} = \exprelset{\G}{} = \ctxrelset{\G}{}$.
\end{theorem}

\begin{proof}
  $\ciurelset{\G}{} = \exprelset{\G}{} \subseteq \ctxrelset{\G}{} \subseteq \ciurelset{\G}{}$
  by Theorems~\ref{thm-ciu-equals-e}, \ref{thm-exp-subset-ctx}, and
  \ref{thm-ctx-subset-ciu}, respectively.
\end{proof}


\section{Contextual Equivalence} \label{sec:equivs}

\begin{definition}\label{def:ctxeq}
  If $\hastype{\G}{e}{\Exptype}$ and $\hastype{\G}{e'}{\Exptype}$, we
  say $e$ and $e'$ are \emph{contextually equivalent} ($e \ctxeq e'$)
  if both ${\ctxrel{\G}{}{e}{e'}}$ and ${\ctxrel{\G}{}{e'}{e}}$.
\end{definition}

In this section we use the machinery from the last few sections to
prove several equivalence schemes. The equivalences fall into three
categories:
\begin{enumerate}
\item provable directly using CIU and Theorem~\ref{thm-ciu=ctx}
\item dependent on ``entropy-shuffling''
\item mathematical properties of \R, probability distributions, etc
\end{enumerate}
Some equivalences of the first and second kinds are listed in
Figure~\ref{figure:axioms}; Section~\ref{section:quasi-denotational} gives
some examples of the third kind.

\begin{figure}
\begin{align*}
  \appexp{\lambdaexpp{x}{e}}{v}
  &~~\ctxeq~~ e[v/x]
  \tag{$\beta_v$}
  \\
  \letexp{x}{v}{e}
  &~~\ctxeq~~ e[v/x]
  \tag{$\letvax$}
  \\
  \letexp{x}{e}{x}
  &~~\ctxeq~~ e
  \tag{$\letidax$}
  \\
  \opexp{\op}{v_1, \cdots, v_n}
  &~~\ctxeq~~ v \quad\text{where $\delta(\op, v_1, \cdots, v_n) = v$}
  \tag{$\delta$}
  \\
  \letexp{x_2}{\letexpp{x_1}{e_1}{e_2}}{e_3}
  &~~\ctxeq~~
  \letexp{x_1}{e_1}{\letexpp{x_2}{e_2}{e_3}}
  \tag{assoc}
  \\
  \letexp{x_1}{e_1}{\letexp{x_2}{e_2}{e_3}}
  &~~\ctxeq~~
  \letexp{x_2}{e_2}{\letexp{x_1}{e_1}{e_3}}
  \tag{commut}
\end{align*}

\medskip 
In (assoc), $x_1 \not\in \itop{FV}(e_3)$.
In (commut), $x_1 \not\in \itop{FV}(e_2)$ and $x_2 \not\in \itop{FV}(e_1)$.
\caption{An incomplete catalog of equivalences}
\label{figure:axioms}
\end{figure}

\subsection{$\beta_v$, $\rmop{let}_v$, and $\delta$}


The proof for $\beta_v$ demonstrates the general pattern of
equivalence proofs using CIU: first we prove the equation holds for
closed expressions, then we generalize to open terms by considering
all closing substitutions.

\begin{lemma} \label{lemma:betav-closed}
  If $\hastype{}{\appexp{\lambdaexpp{x}{e}}{v}}{\Exptype}$, then
  $\appexp{\lambdaexpp{x}{e}}{v} \ctxeq e[v/x]$.
\end{lemma}
\begin{proof}
  By Lemma~\ref{lemma:meas-eqns}, the definition of $\ciurelset{}{}$,
  and Theorem~\ref{thm-ciu=ctx}.
\end{proof}

\begin{corollary}[$\beta_v$] \label{thm:betav}
  If $\hastype{\G}{\appexp{\lambdaexpp{x}{e}}{v}}{\Exptype}$, then
  $\appexp{\lambdaexpp{x}{e}}{v} \ctxeq e[v/x]$.
\end{corollary}
\begin{proof}
  By Lemma~\ref{lemma:betav-closed}, any closed instances of these
  expressions are contextually equivalent and thus
  CIU-equivalent. Hence the open expressions are CIU-equivalent and
  thus contextually equivalent.
\end{proof}

The proofs of $\letvax$ and $\delta$ are similar.


\subsection{Rearranging Entropy}
\label{sec:shuffling}

The remaining equivalences from Figure~\ref{figure:axioms} involve
non-trivial changes to the entropy access patterns of their
subexpressions.
In this section we characterize a class of transformations on the
entropy space that are measure-preserving. In the next section we use
these functions to justify reordering and rearranging subexpression
evaluation.

\begin{definition}[measure-preserving]
A function $\phi : \Ent \to \Ent$ is measure-preserving
when for all measurable $g : \Ent \to \Rplus$,
\begin{displaymath}
  \leb{g(\phi(\s))}{\s} = \leb{g(\s)}{\s}
\end{displaymath}
Note that this definition is implicitly specific to the stock entropy
measure $\mu_\Ent$, which is sufficient for our needs.
\end{definition}

More specifically, the kinds of functions we are interested in are
ones that break apart the entropy into independent pieces using
$\pi_L$ and $\pi_R$ and then reassemble the pieces of entropy using
\cons.  Pieces may be discarded, but no piece may be used more than
once.

For example, the following function is measure-preserving:
\[
  \phi_c(\s_1 \cons (\s_2 \cons \s_3)) = \s_2 \cons (\s_1 \cons \s_3)
\]
Or equivalently, written using explicit projections:
\[
  \phi_c(\s) = \pi_L(\pi_R(\s)) \cons (\pi_L(\s) \cons \pi_R(\pi_R(\s)))
\]
We will use this function in Theorem~\ref{thm-commut} to justify
\ttop{let}-reordering.
%

To characterize such functions, we need some auxiliary definitions:
\begin{itemize}
\item A \emph{path} $p = [d_1 \comdots d_n]$ is a (possibly empty) list
of directions ($L$ or $R$).
It represents a sequence of projections, and it can be viewed as a function from
$\Ent$ to $\Ent$.
\begin{displaymath}
  [d_1 \comdots d_n](\s) = (\pi_{d_1} \circ \dots \circ \pi_{d_n})(\s)
\end{displaymath}
\item A \emph{finite shuffling function} (FSF) $\phi$ is either a path
or $\phi_1 \cons \phi_2$ where $\phi_1$ and $\phi_2$ are FSFs.
It represents the disassembly and reassembly of entropy, and it can be viewed as a
recursively defined function from $\Ent$ to $\Ent$.
  \begin{align*}
    \phi(\s) = \begin{cases}
      p(\s) & \text{if \,} \phi = p \\
      \phi_1(\s) \cons \phi_2(\s) & \text{if \,} \phi = \phi_1 \cons \phi_2
    \end{cases}
  \end{align*}
\item A sequence of paths is said to be \emph{non-duplicating} if no
  path is the suffix of another path in the sequence.
\item An FSF is said to be \emph{non-duplicating} if the sequence of
  paths appearing in its definition is non-duplicating.
\end{itemize}

\begin{lemma}
  \label{lemma-shuffling-iterate} 
  Let $p_1 \comdots p_n$ be a non-duplicating sequence of paths and $g : \Ent^n \to \Rplus$. Then
  \begin{displaymath}
    \leb{g(p_1(\s) \comdots p_n(\s))}{\s} = \iterint{g(\s_1 \comdots \s_n)}{\s_1}{\s_n}
  \end{displaymath}
\end{lemma}

\begin{proof}
  By strong induction on the length of the longest path in the
  sequence, and by the definition of non-duplicating and
  Lemma~\ref{lemma-tonelli} (Tonelli).
\end{proof}

\begin{theorem}
  \label{thm-shuffling-pushforward} 
  If $\phi$ is a non-duplicating FSF then $\phi$ is measure preserving.
\end{theorem}
\begin{proof}
  We need to show that for any $g : \Ent \to \Rplus$,
  \begin{displaymath}
    \leb{g(\phi(\s))}{\s} = \leb{g(\s'')}{\s''}
  \end{displaymath}
  If $\phi$ has paths $p_1 \comdots p_n$, then we can decompose $\phi$ using
  $s : \Ent^n \to \Ent$ such that
  \begin{displaymath}
    \phi(\s) = s(p_1(\s) \comdots p_n(\s))
  \end{displaymath}
  where the $p_i$ are non-duplicating.
  Then by Lemma~\ref{lemma-shuffling-iterate} it is enough to show that
  \begin{displaymath}
    \iterint{g(s(\s_1 \comdots \s_n))}{\s_1}{\s_n} = \leb{g(\s'')}{\s''}
  \end{displaymath}

  We proceed by induction on $\phi$.

  \begin{itemize}
    \item case $\phi = p$. This means that $n = 1$ and $s$ is the identity
      function, so the equality holds trivially.
    \item case $\phi = \phi_1 \cons \phi_2$. If $m$ is the number of paths in $\phi_1$,
      then there must be $s_1 : \Ent^m \to \Ent$ and $s_2 : \Ent^{n-m} \to \Ent$ such that
      \begin{displaymath}
        s(\s_1 \comdots \s_m, \s_{m+1} \comdots \s_n) = s_1(\s_1 \comdots \s_m) \cons s_2(\s_{m+1} \comdots \s_n)
      \end{displaymath}
      We can conclude that
      \begin{align*}
        & \iterint{g(s(\s_1 \comdots \s_n))}{\s_1}{\s_n} \\
        & = \iterint{g(s_1(\s_1 \comdots \s_m) \cons s_2(\s_{m+1} \comdots \s_n))}{\s_1}{\s_n} \\
        & = \dint{g(\s \cons \s')}{\s}{\s'} &\just{IH twice} \\
        & = \leb{g(\s'')}{\s''} &\just{Property~\ref{lemma-entropy-splitting}(4)}
      \end{align*}
    \end{itemize}
\end{proof}


\subsection{Equivalences That Depend on Rearranging Entropy}
\label{sec:commut}

  We first prove a general theorem relating value-preserving
transformations on the entropy space:

\begin{theorem}
  \label{thm-entropy-shuffling1}
  Let $e$ and $e'$ be closed expressions, and let $\phi : \Ent \to \Ent$ be a
  measure-preserving transformation such that for all $\s$, $K$, $\t$,
  and $A$
  \begin{displaymath}
    \eval{\s}{e}{K}{\t}{1}{A} \leq
    \eval{\phi(\s)}{e'}{K}{\t}{1}{A}
  \end{displaymath}
  Then $\ctxrel{}{}{e}{e'}$.
\end{theorem}

\begin{proof}
  Without loss of generality, assume $e$ and $e'$ are closed
  (otherwise apply a closing substitution).  By
  Theorem~\ref{thm-ciu=ctx}, it is sufficient to show that for any
  $K$ and $A$, $\Meas{e}{K}{A} \le \Meas{e'}{K}{A}$.
  We calculate:
  \begin{align*}
    \Meas{e}{K}{A}
    & = {\dint {\eval{\s}{e}{K}{\t}1{A}} {\s} {\t}} \\
    & \le {\dint {\eval{\phi(\s)}{e'}{K}{\t}1{A}} {\s} {\t}} \\
    & = {\dint {\eval{\s}{e'}{K}{\t}1{A}} {\s}{\t}}
    \tag{$\phi$ is measure-preserving} \\
    & = \Meas{e'}{K}{A}
  \end{align*}
\end{proof}

\begin{theorem}
  \label{thm:bigstep-entropy-shuffling1}
  Let $e$ and $e'$ be closed expressions, and let $\phi : \Ent \to \Ent$ be a
  measure-preserving transformation such that for all $v$ and $w$,
  \begin{displaymath}
    {\bigstep{\s}{e}{v}{w}} \implies {\bigstep{\phi(\s)}{e'}{v}{w}}.
  \end{displaymath}
  Then $\ctxrel{}{}{e}{e'}$.  
\end{theorem}

\begin{proof}
  We will use Theorem~\ref{thm-entropy-shuffling1}.  
Assume $\eval{\s}{e}{K}{\t}{1}{A} = r > 0$.
Hence by Theorem~\ref{thm:small-step-big-step}, there exist quantities
$v'$, $w'$, $\s'$, and $\t'$ such that
\begin{displaymath}
  {\nconfig{\s}{e}{K}{\t}{1}} \to^* {\nconfig{\s'}{v'}{\haltcont}{\t'}{r}}
\end{displaymath}
with $v' \in A$.  By Theorem~\ref{thm:small-step-big-step} there exist
$v''$, $\s''$, and $w''$ such that
\begin{displaymath}
  {\bigstep{\s}{e}{v''}{w''}} \hbox{\ and\ }
  {\nconfig{\s''}{v''}{K}{\t}{w''}} \to^* {\nconfig{\s'}{v'}{\haltcont}{\t'}{r}} 
\end{displaymath}

By the assumption of the theorem, we have
${\bigstep{\phi(\s)}{e'}{v''}{w''}}$. 

Therefore, by Theorem~\ref{thm:bigstep-smallstep}, there is a $\s'''$
such that
\begin{displaymath}
    \nconfig{\s'}{e}{K}{t}{1} \to^* \nconfig{\s'''}{v''}{K}{\t}{w''}
  \end{displaymath}

  We claim that $\eval{\phi(s)}{e'}{K}{\t}{1}{A} = r$.  Proceed by
  cases on $K$.  We have ${\nconfig{\s''}{v''}{K}{\t}{w''}}$ $\to^*
  {\nconfig{\phi(s)}{v'}{\haltcont}{\t''}{r}}$. If $K = \haltcont$, this
  reduction must have length 0.  Therefore $v'' = v' \in A$ and $w'' =
  r$, so $\eval{\phi(s)}{e'}{K}{\t}{1}{A} = r$.

  Otherwise assume $K = \letcont{x}{e_3}{K'}$. Then both
  ${\nconfig{\s''}{v"}{K}{t}{w''}}$ and
  ${\nconfig{\s'''}{v"}{K}{t}{w''}}$ take a step to
  ${\nconfig{\pi_L(\t)}{e_3[v''/x]}{K'}{\pi_R(\t)}{w''}}$, so
  $\eval{\phi(s)}{e'}{K}{\t}{1}{A} = \eval{\s}{e}{K}{\t}{1}{A}$ $= r$,
  as desired, thus establishing the requirement of
  Theorem~\ref{thm-entropy-shuffling1}. 
\end{proof}

Now we can finally prove the commutativity theorem promised at the beginning.

\begin{theorem}
  \label{thm-commut}
  Let $e_1$ and $e_2$ be closed expressions, and
  ${\hastype{\setof{x_1,x_2}}{e_0}{\Exptype}}$.  Then
  the expressions
  \begin{displaymath}
    {\letexp{x_1}{e_1}{\letexp{x_2}{e_2}{e_0}}}
  \end{displaymath}
and
  \begin{displaymath}
    {\letexp{x_2}{e_2}{\letexp{x_1}{e_1}{e_0}}}
  \end{displaymath}
are contextually equivalent.
\end{theorem}

\begin{proof}[Proof using big-step semantics]
  Let $e$ and $e'$ denote the two expressions of the theorem.  We will
  use Theorem~\ref{thm:bigstep-entropy-shuffling1} with the function
  $\phi_c(\s_1 \cons (\s_2 \cons \s_3)) = \s_2 \cons (\s_1 \cons \s_3)$,
  which preserves entropy as shown in the
  preceding section.  We will show that if ${\bigstep{\s}{e}{v}{w}}$,
    then ${\bigstep{\phi(\s)}{e}{v}{w}}$.
      
Inverting ${\bigstep{\s}{e}{v}{w}}$, we know there must be a derivation

\begin{mathpar}
  {\inferrule*
    {{\bigstep{\pi_L(\s)}{e_1}{v_1}{w_1}} \\
      {\inferrule*
        { {\bigstep{\pi_L(\pi_R(\s))}{e_2}{v_2}{w_{21}}}
          \\
          {\bigstep{\pi_R(\pi_R(\s))}{e_0[v_1/x_1][v_2/x_2]}{v}{w_{22}}}}
        {\bigstep{\pi_R(\s)}{\letexp{x_2}{e_2}{e_0}}{v}{w_2}}
      }
      }
    {\bigstep{\s}{\letexp{x_1}{e_1}
        {\letexp{x_2}{e_2}{e_0}}}
        {v}{w}}
  }
\end{mathpar}
where $w = w_1 \times w_2 = w_1 \times (w_{21}\times w_{22})$.

Since $e_1$ and $e_2$ are closed, they evaluate to closed $v_1$ and
$v_2$, and so the substitutions $[v_1/x_1]$ and $[v_2/x_2]$
commute. Using that and the associativity and commutativity of
multiplication, we can rearrange the pieces to get
\begin{mathpar}
  {\inferrule*
    {{\bigstep{\pi_L(\pi_R(\s))}{e_2}{v_2}{w_{21}}} \\
      {\inferrule*
        {{\bigstep{\pi_L(\s)}{e_1}{v_1}{w_1}} \\
          {\bigstep{\pi_R(\pi_R(\s))}{e_0[v_2/x_2][v_1/x_1]}{v}{w_{22}}}}
        {\bigstep
          {\pi_L(\s) \cons {\pi_R(\pi_R(\s))}}
          {\letexp{x_1}{e_1}{e_0}}
          {v}
          {w_1 \times w_{22}}}}}
    {\bigstep
      {{\pi_L(\pi_R(\s))} \cons {\pi_L(\s) \cons {\pi_R(\pi_R(\s))}}}
      {\letexp{x_2}{e_2}
        {\letexp{x_1}{e_1}{e_0}}}
      {v}{w}}
    }
  \end{mathpar}
  The entropy in the last line is precisely $\phi(\s)$, so the
  requirement of Theorem~\ref{thm:bigstep-entropy-shuffling1} is
  established.
\end{proof}

Theorem~\ref{thm-commut} can also be proven directly from the
small-step semantics using the interpolation and genericity theorems
(\ref{thm:interpolation} and \ref{thm-generic-conts}) to recover the
structure that the big-step semantics makes explicit.  The proof may be
found in \theappendixref.

\begin{corollary}[Commutativity]
  Let $e_1$ and $e_2$ be expressions such that $x_1$ is not free in
  $e_2$ and $x_2$ is not free in $e_1$.  Then
  \begin{displaymath}
    {\letexpp{x_1}{e_1}{\letexp{x_2}{e_2}{e_0}}}
    ~\ctxeq~
    {\letexpp{x_2}{e_2}{\letexp{x_1}{e_1}{e_0}}}
  \end{displaymath}
\end{corollary}

\begin{proof}
  Same as Corollary~\ref{thm:betav}: since all of the closed instances
  are equivalent by Theorem~\ref{thm-commut}, the open expressions are
  equivalent.
\end{proof}


The proofs of \ttop{let}-associativity and $\letidax$ follow the same
structure, except that associativity uses
$\phi_a((\s_1\cons\s_2)\cons\s_3) = \s_1\cons(\s_2\cons\s_3)$
and $\letidax$ uses $\phi_i(\s_1\cons\s_2) = \s_1$.




\subsection{Quasi-Denotational Reasoning}
\label{section:quasi-denotational}

In this section we give a powerful ``quasi-denotational'' reasoning
tool that shows that if two expressions denote the same measure,
they are contextually equivalent. This allows us to import
mathematical facts about real arithmetic and probability
distributions.

To support this kind of reasoning, we need a notion of measure for a
(closed) expression independent of a program continuation.
We define $\vmeasM{e}$ as a measure over arbitrary \emph{syntactic
  values}---not just real numbers as with $\Meas{e}{K}{-}$.
This measure corresponds directly to the $\mu_e$ of
\citet{CulpepperCobb:2017} and $\sem{e}_{\mathrm{S}}$ of \citet{Borgstrom:2017}.
The definition of $\vmeasfun$ uses a generalization of $\rmop{eval}$
from measurable sets of reals ($A$) to measurable sets of syntactic
values ($V$). This requires a measurable space for syntactic
values; we take the construction of \citet[Figure 5]{Borgstrom:2017}
mutatis mutandis.

\begin{definition}
\begin{align*}
  \vmeas{e}{V} &= \dint{ \eval{\s}{e}{\haltcont}{\t}{1}{V} }{\s}{\t}
  \\
  \eval{\s}{e}{K}{\t}{w}{V} &=
  \begin{cases}
    w' & \begin{array}[t]{l}
      \text{if $\nconfig{\s}{e}{K}{\t}{w} \to^* \nconfig{\s'}{v}{\haltcont}{\t'}{w'}$,} \\
      \text{where $v \in V$}
    \end{array} \\
    0  & \begin{array}[t]{l} \text{otherwise} \end{array}
  \end{cases}
\end{align*}
\end{definition}

Our goal is to relate an expression's measure $\vmeas{e}{-}$ with the
measure of that expression with a program continuation
($\Meas{e}{K}{-}$). Then if two expressions have the same measures, we
can use CIU to show them contextually equivalent.

First we need a lemma about decomposing evaluations. It is easiest to
state if we define the value and weight projections of evaluation:
\begin{align*}
    \ev{\s}{e}{K}{\t} &= \begin{cases}
      v & \text{when $\nconfig{\s}{e}{K}{\t}{1} \to^* \nconfig{\s'}{v}{\haltcont}{\t'}{w}$} \\
      \bot & \text{otherwise}
    \end{cases}
    \\
    \ew{\s}{e}{K}{\t} &= \begin{cases}
      w & \text{when $\nconfig{\s}{e}{K}{\t}{1} \to^* \nconfig{\s'}{v}{\haltcont}{\t'}{w}$} \\
      0 & \text{otherwise}
    \end{cases}
\end{align*}

Note that
\begin{align*}
  \eval{\s}{e}{K}{\t}{1}{V}
  &= I_V(\ev{\s}{e}{K}{\t}) \times \ew{\s}{e}{K}{\t}
  \\
  \vmeas{e}{A}
  &= \Meas{e}{\haltcont}{A} \qquad\text{for $A \in \Sigma_\R$}
\end{align*}
where $I_V$ is the characteristic function of $V$.

\begin{lemma}\label{lemma:ev-ew-split}
  \begin{align*}
    \ev{\s}{e}{K}{\t} &=
    \ev{\s'}{\ev{\s}{e}{\haltcont}{\t'}}{K}{\t}
    \\
    \ew{\s}{e}{K}{\t} &=
    \ew{\s'}{\ev{\s}{e}{\haltcont}{\t'}}{K}{\t} \times \ew{\s}{e}{\haltcont}{\t'}
  \end{align*}
\end{lemma}
\begin{proof}
  By reduction-sequence surgery using Theorem~\ref{thm-generic-conts}.
  Note that the primed variables are dead: $\s'$ because $\rmop{ev}$
  returns a value and $\t'$ because $\haltcont$ does not use its
  entropy stack.
\end{proof}

Next we need a lemma from measure theory:

\begin{lemma} \label{lemma:int-assoc}
  If $\mu$ and $\nu$ are measures and $\nu(A) = \int I_A(f(x)) \times w(x) \;\mu(dx)$, then
  \[
  \smint{g(y)}{\nu(dy)} = \smint{g(f(x)) \times w(x)}{\mu(dx)}
  \]
\end{lemma}
\begin{proof}
  By the pushforward and Radon-Nikodym lemmas from measure theory.
\end{proof}

Now we are ready for the main theorem, which says that $\Meas{e}{K}{-}$ can be
expressed as an integral over $\vmeas{e}{-}$ where $K$ appears only in the
integrand and $e$ appears only in the measure of integration.

\begin{theorem}\label{lemma:meas-vmeas}
  \[
  \Meas{e}{K}{A} = \smiiint{ \eval{\s}{v}{K}{\t}{1}{A} }{\vmeas{e}{dv}}{d\s}{d\t}
  \]
\end{theorem}
\begin{proof}
  By integral calculations and Lemma~\ref{lemma:ev-ew-split}:
  \begin{align*}
    \Meas{e}{K}{A}
    &= \dint{ \eval{\s}{e}{K}{\t}{1}{A} }{\s}{\t}
    \\
    &= \dint{ I_A(\ev{\s}{e}{K}{\t}) \times \ew{\s}{e}{K}{\t} }{\s}{\t}
    \\
    &= \qint{ I_A(\ev{\s}{e}{K}{\t}) \times \ew{\s}{e}{K}{\t} }{\s'}{\t'}{\s}{\t}
    \tag{$\mu_\Ent(\Ent) = 1$}
    \\
    &= \textstyle \iiiint \begin{array}[t]{l}
                 I_A(\ev{\s'}{\ev{\s}{e}{\haltcont}{\t'}}{K}{\t}) \times
                 \ew{\s'}{\ev{\s}{e}{\haltcont}{\t'}}{K}{\t} \\ \quad
                 \times \ew{\s}{e}{\haltcont}{\t'}
                 \;d\s' \,d\t' \,d\s \,d\t
               \end{array}
    \tag{Lemma~\ref{lemma:ev-ew-split}}
    \\
    &= \smiiint{ I_A(\ev{\s'}{v}{K}{\t}) \times \ew{\s'}{v}{K}{\t} }{\vmeas{e}{dv}}{d\s'}{d\t}
    \tag{Lemma~\ref{lemma:int-assoc}}
    \\
    &= \smiiint{ \eval{\s'}{v}{K}{\t}{1}{A} }{\vmeas{e}{dv}}{d\s'}{d\t}
  \end{align*}
\end{proof}

As a consequence, two real-valued expressions are contextually equivalent if
their expression measures agree:

\begin{theorem}[$\vmeasfun$ is quasi-denotational] \label{thm:quasi-denotational}
  If $e$ and $e'$ are closed expressions such that
  \begin{itemize}
  \item $e$ and $e'$ are almost always real-valued---that is,
    $\vmeas{e}{\mathit{Values} - \R} = 0$ and likewise for $e'$---and
  \item for all $A \in \Sigma_\R$, $\vmeas{e}{A} = \vmeas{e'}{A}$
  \end{itemize}
  then $e \ctxeq e'$.
\end{theorem}
\begin{proof}
  The two conditions together imply that $\vmeasM{e} = \vmeasM{e'}$.

  We use Theorem~\ref{thm-ciu=ctx}; we must show $\iciurel{}{}{e}{e'}$
  and $\iciurel{}{}{e'}{e}$.
  Choose a continuation $K$ and a measurable set $A \in \Sigma_\R$. Then
  \begin{align*}
    \Meas{e}{K}{A}
    &= \smiiint{ \eval{\s}{v}{K}{\t}{1}{A} }{\vmeas{e}{dv}}{d\s}{d\t}
    \tag{by Lemma~\ref{lemma:meas-vmeas}}
    \\
    &= \smiiint{ \eval{\s}{v}{K}{\t}{1}{A} }{\vmeas{e'}{dv}}{d\s}{d\t}
    \tag{$\vmeasM{e} = \vmeasM{e'}$}
    \\
    &= \Meas{e'}{K}{A}
    \tag{by Lemma~\ref{lemma:meas-vmeas} again}
  \end{align*}
  The proof of $\iciurel{}{}{e'}{e}$ is symmetric.
\end{proof}

Theorem~\ref{thm:quasi-denotational} allows us to import many useful
facts from mathematics about real numbers, real operations, and
real-valued probability distributions. For example, here are a few
equations useful in the transformation of the linear regression example
from Section~\ref{sec:probprog}:
\begin{itemize}
  \newcommand\xhs{s_0\xspace} \newcommand\xhm{m_0\xspace}
\item $x + y = y + x$
\item $(y + x) - z = x - (z - y)$
\item $(\factorexp{x});~ (\factorexp{y}) = (\factorexp{x \texttt{*} y});~ y$
\item $\opexptt{normalpdf}{x - y; 0, s} = \opexptt{normalpdf}{y; x, s}$
\item The closed-form posterior and normalizer for a normal
  observation with normal conjugate prior~\cite{Murphy07Conj}:
  \begin{align*}
    &
    \begin{array}{l}
      \letexp{m}{\opexptt{normal}{\xhm, \xhs}}{} \\
      \letexp{\_}{\factorexp{\opexptt{normalpdf}{d; m, s}}}{} \\
      m
    \end{array}
    \\
    &=
    \begin{array}{l}
      \letexp{m}{\opexptt{normal}{
          \left( \frac{1}{\xhs^2} + \frac{1}{s^2} \right)^{-1}
          \left( \frac{\xhm}{\xhs^2} + \frac{d}{s^2} \right)
          ,
          \left( \frac{1}{\xhs^2} + \frac{1}{s^2} \right)^{-1/2} }}{} \\
      \letexp{\_}{\factorexp{\opexptt{normalpdf}{d; \xhm, (\xhs^2 + s^2)^{1/2}}}}{} \\
      m
    \end{array}
  \end{align*}
  Note that we must keep the normalizer (the marginal likelihood of
  $d$); it is needed to score the hyper-parameters $\xhm$ and $\xhs$.
\end{itemize}

Section~\ref{sec:fix-dist} contains an additional application of
Theorem~\ref{thm:quasi-denotational}.

\subsection{An Application} \label{section:an-application}

Recall the example program from the introduction and the proposed
transformation:
\par\medskip
\begin{centering}
\begin{minipage}{3in}
  \begin{alltt}
    A = normal(0, 10)
    B = normal(0, 10)
    f(x) = A*x + B
    factor normalpdf(f(2) - 2.4; 0, 1)
    factor normalpdf(f(3) - 2.7; 0, 1)
    factor normalpdf(f(4) - 3.0; 0, 1)
  \end{alltt}
\end{minipage}
$\rightarrow$
\begin{minipage}{3in}
  \begin{alltt}
    A = normal(0, 10)
    factor \(Z\)(A)
    B = normal(\(M\)(A), \(S\)(A))
  \end{alltt}
\end{minipage}
\end{centering}

The core of the transformation is the last equivalence from
Section~\ref{section:quasi-denotational}, which transforms an
observation with a conjugate prior into the posterior and normalizer
(which scores the prior's hyperparameters).
But applying that transformation requires auxiliary steps to focus the
program into the requisite shape:
\begin{itemize}
\item Inline \texttt{f} to expose the dependence of the observations on \texttt{B}.
\item Rewrite the observations to the form \opexptt{normalpdf}{\_; \texttt{B}, \_}
  using facts about arithmetic and \texttt{normalpdf}.
\item Reassociate the (implicit) \texttt{let}s to isolate the
  definition of \texttt{B} and the first observation from the rest of
  the program's main \texttt{let} chain.
\end{itemize}
That sets the stage for the application of the conjugacy
transformation for one observation. Additional shuffling is required
to process subsequent observations. Moreover, each of the mathematical
rewrite rules needs help from the rules of Figure~\ref{figure:axioms} to
manage the intermediate \texttt{let} bindings required by our
language's syntax.

An alternative transformation strategy is to combine the observations
beforehand using equations about products of normal densities. The
same preliminary transformations are necessary, but the
observation-processing loop is eliminated.

\subsection{Other Equivalences}

The list of equivalences presented in this section is not
exhaustive.
On the $\lambda$-calculus side, we focused on a few broadly applicable
rules that involve only syntactic restrictions on the
sub-expressions---specifically, constraints on free variables. There
are other equivalences that require additional semantic
constraints. For example, a \texttt{let}-binding is useless and can be
dropped if the right-hand side has an expression measure of weight 1;
that is, it nearly always terminates without an error and it does not
(effectively) use \texttt{factor}. Similarly, hoisting an
argument-invariant expression out of a function body requires the same
conditions, and the expression must also be deterministic.

On the domain-specific side, Theorem~\ref{thm:quasi-denotational}
works well for programs that contain first-order islands of sampling,
scoring, and mathematical operations. But programs with mathematics
tangled up with higher-order code, it would be necessary to find
either a method of detangling them or a generalization to higher-order
expression measures.

\section{Formally Related Work} 
\label{sec:mismatch}

Our language model differs from other models of probabilistic
languages, such as that of \citet{Borgstrom:2016}, in the following
ways. Our language
\begin{itemize}
\item uses \emph{splitting} rather than \emph{sequenced} entropy,
\item requires \ttop{let}-binding of nontrivial intermediate expressions, and
\item directly models only the standard uniform distribution.
\end{itemize}

These differences, while they make our proofs easier, do \emph{not}
amount to fundamental differences in the meaning of probabilistic
programs. In this section, we show how our semantics corresponds to
other formulations.

\subsection{Splitting versus Sequenced Entropy}

\newcommand\SeqEnt{\ensuremath{\mathbb{T}}}
\newcommand\seqempty{\ensuremath{\epsilon}} 
\newcommand\muH{\ensuremath{\mu_\SeqEnt}}
\newcommand\hs{\ensuremath{t}}
\newcommand\seqto{\ensuremath{\ddot{\to}}}
\newcommand\seqconfig[4]{\langle{#1}\csep{#2}\csep{#3}\csep{#4}\rangle}
\newcommand\seqeval[5]{\ensuremath{{\rmop{e\"val}}({#1},{#2},{#3},{#4},{#5})}}
\newcommand\seqevalfun{\ensuremath{{\rmop{e\"val}}}\xspace}
\newcommand\seqaeval[6]{\ensuremath{{\rmop{e\"val}}^{({#1})}({#2},{#3},{#4},{#5},{#6})}}
\newcommand\SeqMeas[3]{\ensuremath{\ddot{\mu}({#1},{#2},{#3})\xspace}}
\newcommand\MeasFun{\ensuremath{\mu}\xspace}
\newcommand\SeqMeasFun{\ensuremath{\ddot{\mu}}\xspace}
\newcommand\SeqaMeas[4]{\ensuremath{\ddot{\mu}^{(#1)}({#2},{#3},{#4})\xspace}}

Let the sequenced entropy space $\SeqEnt$ be the space of finite
sequences (``traces'') of real numbers~\cite[Section 3.3]{Borgstrom:2017}:
\[
\SeqEnt = \bigcup_{n \geq 0} \R^n
\]
Its stock measure $\muH$ is the sum of the standard Lebesgue measures
on $\R^n$ (but restricted to the Borel algebras on $\R^n$ rather than
their completions with negligible sets). Note that $\muH$ is
infinite.
%

We write $\seqempty$ for the empty sequence and $r\cons\hs$ for the
sequence consisting of $r$ followed by the elements of $\hs$.
Integration with respect to $\muH$ has the following
property:
\[
\smint{ f(\hs) }{\muH(d\hs)}
= f(\seqempty) + \smiint{ f(r\cons\hs) }{\muH(d\hs)}{\lambda(dr)}
\]

We define $\seqto$, $\seqeval{\hs}{e}{K}{w}{A}$, and
$\SeqMeas{e}{K}{A}$\footnote{The dots are intended as a mnemonic for
  sequencing.} as the sequenced-entropy analogues of $\to$,
$\rmop{eval}$, and $\mu$. Here are some representative rules of
$\seqto$:
\begin{align*}
  \seqconfig{\hs}{\letexp{x}{e_1}{e_2}}{K}{w}
  &~\seqto~
  \seqconfig{\hs}{e_1}{\letcont{x}{e_2}{K}}{w}
  \\
  \seqconfig{\hs}{v}{\letcont{x}{e_2}{K}}{w}
  &~\seqto~
  \seqconfig{\hs}{e_2[v/x]}{K}{w}
  \\
  \seqconfig{r\cons\hs}{\sampleexp}{K}{w}
  &~\seqto~
  \seqconfig{\hs}{\const{r}}{K}{w} \quad\text{(when $0 \leq r \leq 1$)}
  \\
  \seqconfig{\hs}{\factorexp{\const{r}}}{K}{w}
  &~\seqto~
  \seqconfig{\hs}{\const{r}}{K}{w \times r} \quad\text{(when $r > 0$)}
\end{align*}
and here are the definitions of $\seqevalfun$ and $\SeqMeasFun$:
\begin{align*}
  \seqeval{\hs}{e}{K}{w}{A} &=
  \begin{cases}
    w' & \text{if $\seqconfig{\hs}{e}{K}{w} \seqto^* \seqconfig{\seqempty}{r}{\haltcont}{w'}$,
               where $r \in A$} \\
    0  & \text{otherwise}
  \end{cases}
  \\
  \SeqMeas{e}{K}{A}
  & = \smint{ \seqeval{\hs}{e}{K}{1}{A} }{\muH(d\hs)}
\end{align*}
Note that an evaluation counts only if it completely exhausts its
entropy sequence $\hs$. The approximants $\seqevalfun^{(n)}$ and
$\SeqMeasFun^{(n)}$ are defined as before; in particular, they are
indexed by number of steps, not by random numbers consumed.

In general, the entropy access pattern is so different between split
and sequenced entropy models that there is no correspondence between
individual evaluations, and yet the resulting measures are equivalent.

\begin{lemma}\label{lemma:seqeval-step}
  If $\seqconfig{\hs}{e}{K}{w} \seqto \seqconfig{\hs'}{e'}{K'}{w'}$, then
  $\seqaeval{p+1}{\hs}{e}{K}{w}{A} = \seqaeval{p}{\hs'}{e'}{K'}{w'}{A}$.
\end{lemma}
\begin{proof}
  By the definition of $\seqevalfun^{(p+1)}$.
\end{proof}

\begin{lemma} \label{lemma:seqmeas-eqns}
  The equations of Lemma~\ref{lemma:meas-eqns} also hold for
  $\SeqMeasFun^{(n)}$ and $\SeqMeasFun$.
\end{lemma}
\begin{proof}
  By definition of $\SeqMeasFun$ and
  Lemma~\ref{lemma:seqeval-step}. In fact, in contrast to
  Lemma~\ref{lemma:meas-eqns}, most of the cases are utterly
  straightforward, because no entropy shuffling is necessary. The
  $\ttop{sample}$ case is different, because it relies on the
  structure of the entropy space:
  \begin{align*}
    \SeqMeas{e}{K}{A}
    &= \smint{ \seqeval{\hs}{\sampleexp}{K}{1}{A} }{d(\hs)}
    \\
    &= \seqeval{\seqempty}{\sampleexp}{K}{1}{A}
    + \smiint{ \seqeval{r \cons \hs}{\sampleexp}{K}{1}{A} }{\muH(dt)}{\lambda(dr)}
    \\
    &= 0 + \smiint{ I_{[0,1]}(r) \times \seqeval{\hs}{\const r}{K}{1}{A} }{\muH(dt)}{\lambda(dr)}
    \\
    &= \smintBB{0}{1}{ \SeqMeas{\const r}{K}{A} }{\lambda(dr)}
  \end{align*}
\end{proof}

\begin{theorem}[$\SeqMeasFun = \MeasFun$] \label{lemma:seqmeas}
  For all $e$, $K$, and $A\in\Sigma_\R$, $\SeqMeas{e}{K}{A} = \Meas{e}{K}{A}$.
\end{theorem}
\begin{proof}
  We first show $\SeqMeasFun^{(n)} = \MeasFun^{(n)}$ by induction on $n$.
  The base case is $\SeqaMeas{0}{e}{K}{A} = \aMeas{0}{e}{K}{A}$. There
  are two subcases: if $e = r$ and $K = \haltcont$, then both results
  are $I_A(r)$. Otherwise, both measures are $0$.
  Lemma~\ref{lemma:seqmeas-eqns} handles the inductive case.
  Finally, since the approximants are pointwise equivalent, their
  limits are equivalent.
\end{proof}

\subsection{Distributions} \label{sec:fix-dist}

The language of \citet{Borgstrom:2016} supports
multiple real-valued distributions with real parameters; sampling from a
distribution, in addition to consuming a random number, multiplies the
current execution weight by the \emph{density} of the distribution at
that point. In this section we show that $\sampleexp$ is equally
expressive, given the inverse-CDF operations.

For each real-valued distribution of interest with $n$ real-valued
parameters, we add the following to the language: a sampling form
$\opexptt{D}{v_1,\dots,v_n}$ and operations $\ttop{Dpdf}$,
$\ttop{Dcdf}$, and $\ttop{Dinvcdf}$ representing the distribution's
density function, cumulative distribution function, and inverse
cumulative distribution function, respectively. The operations take
$n+1$ arguments; by convention we write a semicolon before the
parameters. For example, $\opexptt{gammapdf}{x;k,s}$ represents the
density at $x$ of the gamma distribution with shape $k$ and scale $s$.

We define the semantics of $\ttop{D}$ using the sequenced-entropy
framework by extending $\seqto$ with the following rule schema:
\[
\seqconfig{r\cons\hs}{\opexptt{D}{r_1,\dots,r_n}}{K}{w}
~\seqto~
\seqconfig{\hs}{r}{K}{w \times w'}
\quad\text{where $w' = \opexptt{Dpdf}{r;r_1,\dots,r_n} > 0$}
\]

\begin{theorem}
  $\opexptt{D}{v_1,\dots,v_n}$ and $\opexptt{Dinvcdf}{\sampleexp;v_1,\dots,v_n}$
  are CIU-equivalent (and thus contextually equivalent).
\end{theorem}
\begin{proof}
  By Theorem~\ref{thm:quasi-denotational}. Both expressions are
  real-valued. We must show that their real measures are equal.  We
  abbreviate the parameters as $\vec{v}$. The result follows from the
  relationship between the density function and the cumulative density
  function.
  \begin{align*}
    \vmeas{\opexptt{Dinvcdf}{\sampleexp; \vec{v}}}{A}
    &= \smintBB{0}{1}{ I_A(\opexptt{Dinvcdf}{x;\vec{v}}) }{dx}
    \\
    \intertext{We change the variable of integration with
      $x = \opexptt{Dcdf}{t;\vec{v}}$ and $\frac{dx}{dt} = \opexptt{Dpdf}{t;\vec{v}}$:}
    &= \smintBB{-\infty}{\infty}{ I_A(\opexptt{Dinvcdf}{\opexptt{Dcdf}{t;\vec{v}};\vec{v}}) \times
                \opexptt{Dpdf}{t;\vec{v}} }{dt}
    \\
    &= \smintBB{-\infty}{\infty}{ I_A(t) \times \opexptt{Dpdf}{t;\vec{v}} }{dt}
    \\
    &= \vmeas{\opexptt{D}{\vec{v}}}{A}
  \end{align*}
\end{proof}


\subsection{From let-Style to Direct-Style}

\newcommand\langL{\ensuremath{\mathcal{L}}\xspace}
\newcommand\langD{\ensuremath{\mathcal{D}}\xspace}

\newcommand\tr[1]{\ensuremath{\rmop{tr}\sem{#1}}}
\newcommand\Dconfig[3]{\langle{#1}\csep{#2}\csep{#3}\rangle}
\newcommand\Dto{\to_\langD}
\newcommand\evalDfun{\ensuremath{{\rmop{eval}_\langD}}\xspace}
\newcommand\evalD[4]{\ensuremath{{\rmop{eval}_\langD}({#1},{#2},{#3},{#4})}}
\newcommand\aevalD[5]{\ensuremath{{\rmop{eval}_\langD}^{({#1})}({#2},{#3},{#4},{#5})}}

\newcommand\MeasD[2]{\ensuremath{\mu_\langD({#1},{#2})\xspace}}
\newcommand\MeasDFun{\ensuremath{\mu_\langD}\xspace}
\newcommand\aMeasD[3]{\ensuremath{\mu_\langD^{(#1)}({#2},{#3})\xspace}}

\newcommand\Hole{\ensuremath{[~]}\relax}

Let us call the language of Section~\ref{sec:syntax} \langL (for
``let'') and the direct-style analogue \langD (for ``direct'').
Once again following \citet{Borgstrom:2016}, we give the semantics of
\langD using a CS-style abstract machine, in contrast to the CSK-style
machines we have used until now~\cite{SEwPR}.

Here are the definitions of expressions and evaluation contexts for
\langD:
\begin{align*}
  e &::= v \alt \letexp{x}{e}{e} \alt \appexp{e}{e}
         \alt \opexp{\op}{e, \dots, e} \alt \ifexp{e}{e}{e}
  \\
  E &::= \Hole \alt \letexp{x}{E}{e}
         \alt \appexp{E}{e} \alt \appexp{v}{E}
         \alt \opexp{op}{v,\dots,E,e,\dots}
         \alt \ifexp{E}{e}{e}
\end{align*}
Here are some representative rules for its abstract machine:
\begin{align*}
  \Dconfig{\hs}{E[\appexp{\lambdaexpp{x}{e}}{v}]}{w}
  &\Dto \Dconfig{\hs}{E[e[v/x]]}{w}
  \\
  \Dconfig{r\cons\hs}{E[\sampleexp]}{w}
  &\Dto \Dconfig{\hs}{E[\const{r}]}{w}
\end{align*}
And here are the corresponding definitions of evaluation and measure:
\begin{align*}
  \evalD{\hs}{e}{w}{A}
  &= \begin{cases}
    w' & \text{if $\Dconfig{\hs}{e}{w} \Dto^* \Dconfig{\seqempty}{r}{w'}$, where $r \in A$} \\
    0  & \text{otherwise}
  \end{cases}
  \\
  \MeasD{e}{A}
  &= \smint{ \evalD{\hs}{e}{1}{A} }{\muH(d\hs)}
\end{align*}

To show that our \langL corresponds with \langD, we define a
translation $\tr{-} : \langD \to \langL$. More precisely, $\tr{-}$
translates \langD-expressions to \langL-expressions, such as
\begin{align*}
  \tr{r}
  &= r
  \\
  \tr{\lambdaexp{x}{e}}
  &= \lambdaexp{x}{\tr{e}}
  \\
  \tr{\appexp{e_1}{e_2}}
  &= \letexp{x_1}{\tr{e_1}}{\letexp{x_2}{\tr{e_2}}{\appexp{x_1}{x_2}}}
  \\
  \tr{\letexp{x}{e_1}{e_2}}
  &= \letexp{x}{\tr{e_1}}{\tr{e_2}}
  \\
\intertext{and it translates \langD-evaluation contexts to \langL-continuations, such as}
  \tr{\Hole}
  &= \haltcont
  \\
  \tr{E[\appexp{\Hole}{e_2}]}
  &= \letcont{x_1}{\letexp{x_2}{e_2}{\appexp{x_1}{x_2}}}{\tr{E}}
  \\
  \tr{E[\appexp{v_1}{\Hole}]}
  &= \letcont{x_2}{\appexp{v_1}{x_2}}{\tr{E}}
  \\
  \tr{E[\letexp{x}{\Hole}{e}]}
  &= \letcont{x}{e}{\tr{E}}
\end{align*}

Now we demonstrate the correspondence of evaluation and then lift it to measures.

\newcommand\DCSKconfig[4]{\seqconfig{#1}{#2}{#3}{#4}}
\newcommand\DCSKto{\ensuremath{\to_{\langD \textsc{csk}}}}

\begin{lemma}[Simulation] \label{lemma:simD}
  $\evalD{\hs}{E[e]}{w}{A} = \seqeval{\hs}{\tr{e}}{\tr{E}}{w}{A}$
\end{lemma}
\begin{proof}
  From the CS-style machine above we can derive a corresponding CSK
  machine (call it $\DCSKto$); the technique is standard~\cite{SEwPR}.
  Then it is straightforward to show that
  \[
  \DCSKconfig{\hs}{e}{K}{w} \DCSKto \DCSKconfig{\hs'}{e'}{K'}{w}
  \implies
  \seqconfig{\hs}{\tr{e}}{\tr{K}}{w} \seqto^* \seqconfig{\hs'}{\tr{e'}}{\tr{K'}}{w'}
  \]
  and thus the evaluators agree.
\end{proof}  

\begin{theorem}
  $\MeasD{E[e]}{A} = \Meas{\tr{e}}{\tr{E}}{A}$
\end{theorem}
\begin{proof}
  By definition of $\MeasDFun$ and Lemmas~\ref{lemma:seqmeas} and \ref{lemma:simD}.
\end{proof}

The equational theory for \langD is the \emph{pullback} of the \langL
equational theory over $\tr{-}$.
Compare with \citet{SabryFelleisen:1993}, which explores the pullback
of $\boldsymbol{\lambda}\beta\eta$ and related calculi over the
call-by-value CPS transformation.
For our language \langD, associativity and commutativity combine to
yield a generalization of their $\beta_\mathit{flat}$ and
$\beta_\Omega'$ equations to ``single-evaluation'' contexts $S$:

\[
\begin{altgrammar}
  S &::=& \Hole
  \alt \appexp{S}{e} \alt \appexp{e}{S}
  \alt \letexp{x}{S}{e} \alt \letexp{x}{e}{S}
  \\ &&
  \alt \opexp{\op}{e, \dots, S, e, \dots}
  \alt \ifexp{S}{e}{e}
\end{altgrammar}
\]

\begin{equation}
  \letexp{x}{e}{S[x]} ~\ctxeq~ S[e]
  \qquad\text{when $x \not\in \itop{FV}(S) \cup \itop{FV}(e)$}
  \tag{$\letsax$}
\end{equation}

\section{Informally Related Work}
\label{sec:related-work}

Our language and semantics are based on that of
\citet{CulpepperCobb:2017}, but unlike that language, which is
simply-typed, ours is untyped and thus has recursion and
nonterminating programs. Consequently, our logical relation must use
step-indexing rather than type-indexing.
Using an untyped language instead of a typed one not only introduces
recursion; it increases the universe of expressions to which the
theory applies, but it also makes the equivalence stricter, since the
untyped language admits both more expressions and more contexts.

The construction of our logical relation follows the tutorial of
\citet{Pitts:2010} on the construction of biorthogonal,
step-indexed~\citep{Ahmed:2006} logical relations. Instead of
termination, we use the program measure as the observable behavior,
following \citet{CulpepperCobb:2017}. But unlike that work, where the
meaning of an expression is a measure over arbitrary syntactic values,
we define the meaning of an expression and continuation together
(representing a whole program) as a measure over the reals.  This
allows us to avoid the complication of defining a relation on
measurable sets of syntactic values~\cite[the $\mathcal{A}$
  relation]{CulpepperCobb:2017}.

There has been previous work on contextual equivalence for
probabilistic languages with only \emph{discrete} random variables. In
particular, \citet{BizjakBirkedal:2015} define a step-indexed,
biorthogonal logical relation whose structure is similar to ours,
except that they sum where we integrate, and they use the probability
of termination as the basic observation whereas we compare measures.
Others have applied bisimulation
techniques~\citep{Crubille:2014,Sangiorgi:2016} to languages with
discrete choice; \citet{Erhard:2014} have constructed
fully abstract models for PCF with discrete probabilistic choice using
probabilistic coherence spaces.

\citet{Staton:2016} gives a denotational semantics for a higher-order,
typed language with continuous random variables, scoring, and
normalization but without recursion. Using a variant of that
denotational semantics, \citet{Staton:2017} proves the soundness of
the $\ttop{let}$-reordering transformation for a first-order language.


\begin{acks}

  This material is based upon work sponsored by the
  \grantsponsor{G1}{Air Force Research Laboratory (AFRL)}{} and the
  \grantsponsor{G2}{Defense Advanced Research Projects Agency (DARPA)}{}
  under \grantnum{G1}{Contract No.\ FA8750-14-C-0002}{}.
  The views expressed are those of the authors and do not reflect the official
  policy or position of the Department of Defense or the U.S.  Government.

  This project has received funding from the European
  Research Council (ERC) under the \grantsponsor{G3}{European Union's Horizon 2020
  research and innovation programme}{} (grant agreement \grantnum{G3}{No.\ 695412}{}).
  
\end{acks}

\bibliography{refs}

\appendix

\section{Additional Proofs} \label{sec:appendix-proofs}

\subsection{From Section~\ref{sec:op-sem}}

\begin{proof}[Proof of \ref{lemma:continuations-are-continuous}]
  %
  By induction on $i$.  At $i = 1$, property (b) is true.  So assume
  the proposition at $i$, and check it at $i+1$.   If property (a)
  holds for some $j \le i$, then (a) holds at $i+1$.   Otherwise,
  assume that property (b) holds at $i$, that is,
  assume that $(K_i,\s_i) \succeq (K_1,\s_1)$.

  If $e_i$ is not a value, then either $(K_{i+1},\t_{i+1}) =
  (K_i,\t_i)$, or else $K_{i+1} = \letcont{x}{e}{K_i}$ and $\t_{i+1} =
  \s' \cons \t_i$ for some $x$, $e$ and $\s'$,  in which case
  $(K_{i+1},\t_{i+1}) \succeq (K_1,\t_1)$ by Rule 2 above.  So the
  property holds at $i+1$.

  If $e_i$ is a value, then consider the last step in the derivation
  of $(K_i,\t_i) \succeq (K_1,\t_1)$.  If the last step was Rule 1,
  then property (a) holds at $i$ and therefore it holds at $i+1$.
  If the last step was Rule 2, then
  $(K_i, \t_i) = (\letcont{x}{e}{K'}, \s' \cons \t')$ for some $\s'$,
  where $(K', \t') \succeq (K_1, \t_1)$.
  So the configuration at step $i$ is a
  return from a \texttt{let}, and
  $(K_{i+1}, \t_{i+1}) = (K', \t') \succeq (K_1,\s_1)$
  by inversion on Rule 2, so again the property
  holds at $i+1$.
\end{proof}

\subsection{From Section~\ref{sec:measures}}

\begin{proof}[Proof of \ref{lemma-eval-measurable} (Measurability)]
The argument goes as follows.   Following \citet[Figure
5]{Borgstrom:2016}, turn the set of expressions and
continuations into a metric space by setting
$d(\const r,\const {r'}) = \lvert r-r'\rvert$;
$d(\appexp{e_1}{e_2}, \appexp{e'_1}{e'_2} = d(e_1,e_1')+d(e_2,e_2')$,
etc., setting $d(e,e') = \infty$ if $e$ and $e'$ are not the same up
to constants. Extend this to become a measurable space on
configurations by constructing the product space, using the Borel sets
for the weights and the natural measurable space on the entropy
components.  Note that in this space, singletons are measurable sets.

The next-configuration function $\nextstep : \Config \to \Config$ is
measurable; the proof follows the pattern of \citet[Lemmas
  72--84]{Borgstrom:2017}.
%
It follows that the $n$-fold composition of \nextstep,
$\nextstep^{(n)}$ is measurable, as is $\finishcomp \circ
\nextstep^{(n)}$, where $\finishcomp$ extracts the weight of halted
configurations.

Now we consider the measurability of $\rmop{eval}$.
Let $B$ be a Borel set in the reals and set
\begin{align*}
  C &=\setof{(\s, e, K, \t, w) \suchthat {\eval{\s}{e}{K}{\t}{w}{A}} \in B} \\
    &= \bigcup_n((\finishcomp \circ \nextstep^{(n)})^{-1}(B))
\end{align*}
Since $C$ is equal the countable union of measurable sets, it is
measurable, and thus $\rmop{eval}$ is measurable with respect to the
product space of all of its arguments.
To show $\eval{\s}{e}{K}{\t}{w}{A}$ is measurable with respect to $(\s, \t)$
for any fixed $e$, $K$, $w$, and $A$, we note that
\[
\left( \s,\t \mapsto \eval{\s}{e}{K}{\t}{w}{A} \right)
= \rmop{eval} \circ \left(\s,\t \mapsto (\s, e, K, \t, w, A)\right)
\]
The function $\left(\s,\t \mapsto (\s, e, K, \t, w, A)\right)$ is
measurable, as it is just a product of constant and identity
functions. Thus the composition is measurable.
\end{proof}

\begin{proof}[Proof for \ref{lemma:meas-eqns} (\sampleexp)]
  \begin{align*}
    & {\aMeas {p+1} {\sampleexp} {K}{A}} \\
    & = {\dint {\aeval {p+1} {\s} {\sampleexp} {K} {\t} 1{A}}
               {\s} {\t}} \\
    & = {\dint {\aeval {p} {\pi_R{\s}}
                       {\const {\pi_U(\pi_L(\s))}}
                       {K}
                       {\t}
                       1{A}}
               {\s} {\t}}
    \tag{Lemma~\ref{lemma-opsem-trivial}} \\
    & = {\tint {\aeval {p} {\s_{2}}
                       {\const {\pi_U(\s_{1}) }}
                       {K}
                       {\t}
                       1{A}}
               {\s_{1}} {\s_{2}} {\t}}
    \tag{Property~\ref{lemma-entropy-splitting}.4} \\
    & = {\tint {\aeval {p} {\s_{2}}
                       {\const {\pi_U(\s_{1}) }}
                       {K}
                       {\t}
                       1{A}}
               {\s_{2}} {\t} {\s_{1}}}
    \tag{Lemma~\ref{lemma-tonelli} twice} \\
    & = {\leb {\aMeas {p} {\const {\pi_U(\s_{1})}} {K}{A}} {\s_{1}}} \\
    & = \smintBB{0}{1}{\aMeas{p}{\const{r}}{K}{A}}{dr}
    \tag{Property~\ref{lemma-entropy-def}.2}
  \end{align*}
\end{proof}


\subsection{From Section~\ref{sec:log-rel}}

\begin{proof}[Proof of \ref{lemma:compatibility} (Variables)]
  We must show that for all $n$ and $(\g, \g') \in
  \ityenvrelset{\G}{n}$, $(\g(x),\g'(x)) \in \ivalrelset{}{n}$.  But that is
  true by the definition of $\ityenvrelset{\G}{n}$.
\end{proof}

\begin{proof}[Proof of \ref{lemma:compatibility} ($\l$)]
  Without loss of generality, assume
  $x \not\in \G$, and hence for any $\g$, $(\lambdaexp{x}{e})\g =
  \lambdaexp{x}{e\g}$.
  We must show, for all $n$, if $\ityenvrel{n}{\G}{\g}{\g'}$, then
  ${\ivalrel{}{n}{\lambdaexp{x}{e\g}}{\lambdaexp{x}{e'\g'}}}$.

  Following the definition of ${\ivalrelset{}{n}}$, choose $m < n$ and
  ${\ivalrel{}{m}{v}{v'}}$.  We must show that
  ${\iexprel{}{m}{e\g[v/x]}{e'\g'[v'/x]}}$.
  Since $m < n$,
  we have $\ityenvrel{m}{\G}{\g}{\g'}$ and ${\ivalrel{}{m}{v}{v'}}$.
  Therefore  $\ityenvrel{m}{\G,x}{\g[v/x]}{\g'[v'/x]}$, so
  ${\iexprel{}{m}{e\g[v/x]}{e'\g'[v'/x]}}$ by the definition of ${\iexprelset{}{m}}$.
\end{proof}

\begin{proof}[Proof of \ref{lemma:compatibility} (Value-compatibility implies expression-compatibility)]
  Choose $n$ and $(\g, \g') \in \ityenvrelset{\G}{n}$, so we have
  $\valrel{}{n}{v\g}{v'\g'}$.  We must show that
  $\exprel{}{n}{v\g}{v'\g'}$.

  Following the definition of $\iexprelset{}{n}$, choose $m \leq n$,
  $(K, K') \in \icontrelset{}{m}$, and $A$.
  Since $m \le n$, we have
  $\valrel{}{m}{v\g}{v'\g'}$, so $\aMeas{m}{v}{K}{A} \le
  \Meas{v'}{K'}{A}$.  Since we have this for all $m \le n$, we conclude that
  $\exprel{}{n}{v\g}{v'\g'}$.
\end{proof}

\begin{proof}[Proof of \ref{lemma:compatibility} (application)]
  Choose $n$, and assume $\ityenvrel{n}{\G}{\g}{\g'}$.  Then
  ${\ivalrel{}{n}{v_1\g\,}{v_1'\g'}}$ and \\
  $\ivalrel{}{n}{v_2\g}{v_2'\g'}$.
  We must show
  ${\iexprel{}{n}{\appexp{v_1\g\,}{v_2\g}}{\appexp{v_1'\g'}{v_2'\g'}}}$

  If $v_1\g$ is of the form $\const r$, then $\aMeas{m}{v_1\g}{K}{A} = 0$
  for any $m$, $K$, and $A$, so by Lemma~\ref{lemma:misc-properties} the
  conclusion holds.

  Otherwise, assume $v_1\g$ is of the form \lambdaexp{x}{e}, and so
  $v_1'\g'$ is of the form \lambdaexp{x}{e'}.
  So choose $m \le n$ and $A$, and let  $\contrel{}{m}{K}{K'}$.
  We must show that
  \[\aMeas {m} {\appexp{\lambdaexp{x}{e\g}}{v_2\g}} {K}{A}
  \le \Meas {\appexp{\lambdaexp{x}{e'\g'}}{v_2'\g'}}{K'}{A}.\]

  If $m = 0$ the left-hand side is 0.  So assume $m \ge 1$.  Since all
  the relevant terms are closed and the relations on closed terms are
  antimonotonic in the index, we have
  ${\ivalrel{}{m}{\lambdaexp{x}{e\g}}{\lambdaexp{x}{e'\g'}}}$ and
  ${\ivalrel{}{m-1}{v_1\g\,}{v_1'\g'}} $.  Therefore
  ${\iexprel{}{m-1}{e\g[v_2\g/x]}{e'\g'[v_2'\g'/x]}}$.

  Now, $\configsw{\appexp{\lambdaexp{x}{e\g}\,}{v_2\g}}{K} \to
  \configsw{e\g[v_2\g/x]}{K}$, and similarly for the primed side.  So we
  have
  \begin{align*}
    \aMeas {m} {\appexp{\lambdaexp{x}{e\g}\,}{v_2\g}}{K}{A}
    &= \aMeas {m-1} {e\g[v_2\g/x]}{K}{A}  \tag{Lemma~\ref{lemma-appexp}}
    \\ &\le \Meas {e'\g'[v_2'\g'/x]}{K'}{A}  \tag{by
      ${\iexprel{}{m-1}{e\g[v_2\g/x]}{e'\g'[v_2'\g'/x]}}$}
    \\ &= \Meas {\appexp{\lambdaexp{x}{e'\g'}\,}{v_2'\g'}}{K'}{A}
  \end{align*}
\end{proof}

\begin{proof}[Proof of \ref{lemma:compatibility} (operations)]
  Choose $n$, $\ityenvrel{n}{\G}{\g}{\g'}$, $m \leq n$,
  $(K, K') \in \icontrelset{}{m}$, and $A$.

  Since the arguments are related, for each $i$ either $v_i\g$ and
  $v_i'\g'$ are the same real number $\const{r_i}$ or both are
  closures. So either
  \[
  \deltafun(\op^k, v_1\g, \dots, v_k\g) = \deltafun(\op^k, v_1'\g', \dots, v_k'\g')
  = \const{r}
  \]
  or both are undefined, in which case both measures are 0. Assuming
  the result is defined and $m > 0$,
  \begin{align*}
    & \aMeas{m}{\opexp{\op^k}{v_1\g, \dots, v_k\g}}{K}{A}
    \\ &= \aMeas{m-1}{\const{r}}{K}{A} \tag{Lemma~\ref{lemma-opexp}}
    \\ &\leq \Meas{\const{r}}{K'}{A} \tag{by definition of $\icontrelset{}{m}$}
    \\ &= \Meas{\opexp{\op^k}{v_1'\g', \dots, v_1'\g'}}{K'}{A}
  \end{align*}
\end{proof}


\begin{proof}[Proof of \ref{lemma:compatibility} ($\haltcont$)]
  We must show that for any $m$ and any $(v,v') \in \ivalrelset{}{n}$,
  $\aMeas{m}{v}{\haltcont}{A} \le \Meas{v'}{\haltcont}{A}$.
  But $\ivalrel{}{n}{v}{v'}$ implies either $v = v' = \const r$ or
  both $v$ and $v'$ are \l-expressions, in which case both sides of
  the inequality are 0.
\end{proof}


\begin{proof}[Proof of \ref{lemma:compatibility} (continuations)]
Choose $n$, $m \leq n$, $(v,v') \in \ivalrelset{}{m}$, and $A \in \Sigma_\R$.
We need to show
$\aMeas{m}{v}{\letcont{x}{e}{K}}{A} \leq \Meas{v'}{\letcont{x}{e'}{K'}}{A}$.
Assume $m > 0$, otherwise trivial. 

By Lemma~\ref{lemma:meas-eqns}, the left-hand side is $\aMeas{m-1}{e[v/x]}{K}{A}$
and the right-hand side is $\Meas{e'[v'/x]}{K}{A}$.
The inequality follows from $\exprel{\{x\}}{}{e}{e'}$.
\end{proof}

To establish compatibility for \ttop{let}, we need some finer information:

\begin{lemma}
  \label{lemma:pitts-4.2}
  Given $n$, and $e$, $e'$ with a single free variable $x$, with the
  property that
  \begin{displaymath}
    (\forall p \le n)(\forall v,v')
    [\ivalrel{}{p}{v}{v'} \implies \iexprel{}{p}{e[v/x]}{e'[v'/x]}]
  \end{displaymath}
  Then for all $m \le n$,
  \begin{displaymath}
    \icontrel{}{m}{K}{K'} \implies \icontrel{}{m}{\letcont{x}{e}{K}}{\letcont{x}{e'}{K'}}
  \end{displaymath}
\end{lemma}
\begin{proof}
  Choose $m \le n$ and $(K,K') \in \icontrelset{}{m}$. To show
  $\icontrel{}{m}{\letcont{x}{e}{K}}{\letcont{x}{e'}{K'}}$, choose
  $p \le m$, $(v,v')\in \ivalrelset{}{p}$, and $A$.

  We must show
  $\aMeas{p}{v}{\letcont{x}{e}{K}}{A} \le \Meas{v'}{\letcont{x}{e'}{K'}}{A}$.

  If $p = 0$, the result is trivial.  So assume $p > 0$ and calculate:
  \begin{align*}
    & \aMeas{p}{v}{\letcont{x}{e}{K}}{A}
    \\& = \aMeas{p-1}{e[v/x]}{K}{A}  \tag{Lemma~\ref{lemma-return}}
    \\& \le \Meas {e'[v'/x]} {K'}{A}
    \\& = \Meas{v'}{\letcont{x}{e'}{K'}}{A}
  \end{align*}
  where the inequality follows from
  $(v,v') \in \ivalrelset{}{p} \subseteq \ivalrelset{}{p-1}$ and
  $(K,K') \in \icontrelset{}{m} \subseteq {\contrelset{}{p}} \subseteq \icontrelset{}{p-1}$.
\end{proof}

Now we can prove compatibility under \ttop{let}.

\begin{proof}[Proof of \ref{lemma:compatibility} (\ttop{let})]
  Choose $n$ and  $(\g, \g') \in \ityenvrelset{\G}{n}$.  So we have
  ${\iexprel{}{m}{e_1\g}{e'_1\g'}}$ for all $m \le n$.

  \sloppy Furthermore, if $m \le n$ and $\ivalrel{}{m}{v}{v'}$, then
  ${\ityenvrel{\G,x}{m}{\g[x:=v]}{\g'[x:=v']}}$.  Therefore
  ${\iexprel{}{m}{e_2\g[x:=v]}{e'_2\g'[x:=v']}}$.   So $(e_2\g, e'_2\g')$
  satisfies the hypothesis of Lemma~\ref{lemma:pitts-4.2}.

  So choose $m \le n$ and $(K,K') \in \icontrelset{}{m}$.  Without
  loss of generality, assume $m > 0$. Then by Lemma~\ref{lemma:pitts-4.2} we have
  \begin{equation}
    {\icontrel{}{m}
      {\letcont{x}{e\g}{K}}
      {\letcont{x}{e'\g'} {K'}}}\label{eq:aa}
  \end{equation}
  Choose $A$. Now we can calculate:
  \begin{align*}
    {\aMeas{m}{\letexp{x}{e_1\g}{e_2\g}}{K}{A}}
    & = {\aMeas{m-1}{e_1\g}{\letcont{x}{e_2\g}{K}}{A}}  \tag{Lemma~\ref{lemma-let}}
    \\& \le {\Meas {e_1'\g'}{\letcont{x}{e_2'\g'}{K'}}{A}}
    \\& = {\Meas {\letexp{x}{e_1'\g'}{e_2\g'}}{K}{A}}
  \end{align*}
  where the  inequality follows from
  ${\iexprel{}{m}{e_1\g}{e'_1\g'}}$ and (\ref{eq:aa}).
\end{proof}

\begin{proof}[Proof of \ref{lemma:compatibility} (\ttop{if})]
  Choose $n$, $(\g, \g') \in \ityenvrelset{\G}{n}$, $m \leq n$,
  $(K, K') \in \icontrelset{}{m}$, and $A \in \SigmaR$.
  Assume that $m > 0$, otherwise the result follows trivially.

  Suppose $v\g = v'\g' = \const r$, and if $r > 0$. Then
  \begin{align*}
    \aMeas{m}{\ifexp{v\g}{e_1\g}{e_2\g}}{K}{A}
    &=    \aMeas{m-1}{e_1\g}{K}{A}         \tag{Lemma~\ref{lemma-ifexp}}
    \\ &\leq \Meas{e_1'\g}{K'}{A}
    \\ &=    \Meas{\ifexp{v'\g'}{e_1'\g'}{e_2'\g'}}{K'}{A}
  \end{align*}
  Likewise for $r \leq 0$ and $e_2, e_2'$.

  Otherwise, neither $v\g$ nor $v'\g'$ is a real constant, and both
  expressions are stuck and have measure 0.
\end{proof}

Everything so far is just an adaptation of the deterministic case.
Now we consider our two effects.

\begin{proof}[Proof of \ref{lemma:compatibility} (\ttop{factor})]
Choose $n$ and $\ityenvrel{\G}{n}{\g}{\g'}$.  We must show 
${\iexprel{}{n}{\factorexp{v\g}} {\factorexp{v'\g'}}}$.  Since
${\ivalrel{\G}{}{v}{v'}}$, it must be that
${\ivalrel{}{n}{v\g}{v'\g'}}$.  So either $v\g = v'\g' = \const r$ for
  some $r$, or $v\g$ is a lambda-expression.

Assume $v\g = v'\g' = \const r$ for some $r > 0$. Choose $1 \le m \le n$, 
${\icontrelalt{}{m}{K}{K'}}$, and $A \in \SigmaR$.  Then we have
\begin{align*}
  {\aMeas {m} {\factorexp{\const r}} {K}{A}}
  &= r \times \aMeas{m-1}{\const r}{K}{A}   \tag{Lemma~\ref{lemma-factor}}
  \\ &\le r \times \Meas{\const r}{K'}{A}
  \\ &= {\Meas {\factorexp {\const r}} {K'}{A}}
\end{align*}
So ${\iexprel{}{n}{\factorexp{\const r}}{\factorexp{\const r}}}$ as desired.

If $v\g$ is $\const r$ for $r \leq 0$ or a lambda-expression, then ${\factorexp{v\g}}$ is stuck, so
${\aMeas {m} {\factorexp{\const r}} {K}{A}} = 0$ and the desired result
holds again.
\end{proof}


\begin{proof}[Proof of \ref{lemma:compatibility} (\sampleexp)]
It will suffice to show that for all $m$, $\icontrelalt{}{m}{K}{K'}$,
and $A \in \SigmaR$,
\begin{displaymath}
  \aMeas {m}{\sampleexp}{K}{A} \le \Meas{\sampleexp}{K'}{A}
\end{displaymath}

At $ m = 0$, the left-hand side is 0 and the result is trivial.
If $m > 0$, then
\begin{align*}
  \aMeas m {\sampleexp}{K}{A}
  &= {\lebs {\aMeas {m-1} {\const {\pi_U(\s)}} {K}{A}}}    \tag{Lemma~\ref{lemma-sample}}
  \\ &\le {\leb {\Meas {\const {\pi_U(\s)}} {K'}{A}} {\s}}
  \\ &= \Meas {\sampleexp}{K'}{A}
\end{align*}
\end{proof}

\subsection{From Section~\ref{sec:ciu-ctx}}

\begin{proof}[Proof of \ref{lemma-swap}]
  Let $K_1$ denote {\letcont{z}{\appexp{z}{v}}{K}}, and
  let $K_1'$ denote {\letcont{z}{\appexp{z}{v'}}{K'}}.  To show
  $\icontrel{}{n+2}{K_1}{K'_1}$, choose $2 \le m \le n+2$, 
  $\ivalrel{}{m}{u}{u'}$, and $A \in \SigmaR$.  We must show
  \begin{displaymath}
    \aMeas {m}{u}{K_1}{A} \le \Meas {u'}{K_1'}{A}
  \end{displaymath}

  There are two possibilities for $\ivalrel{}{m}{u}{u'}$:

  1. $u = u' = \const r$.  Then ${\nconfig{\s}{\const r}{K_1}{\t}{w}} \to
  {\nconfig{\s}{\appexp{\const r}{v}}{K}{\t}{w}}$, which is stuck, so
  $\aMeas {m}{u}{K_1}{A} = 0 \le \Meas {u'}{K_1'}{A}$.

  2. $u = \lambdaexp{x}{e}$ and $u' = \lambdaexp{x}{e'}$ where for all
  $p < m$ and all $\ivalrel{}{p}{u_1}{u_1'}$,
  $\iexprel{}{p}{e[u_1/x]}{e'[u_1'/x]}$.

  Now, for any \s and $w$, we have
  \begin{align*}
    & \configsw{\lambdaexp{x}{e}}{K_1}  \\
    & \to \configsw{\appexp{\lambdaexp{x}{e}}{v}}{K} \\
    & \to {\configsw {e[v/x]} {K}}
  \end{align*}
  so $\aMeas{m}{\lambdaexp{x}{e}}{K_1}{A} = \aMeas{m-2}{e[v/x]}{K}{A}$, and
  similarly on the primed side (but with $\Meas{-}{-}{-}$ in place of
  $\aMeas{m}{-}{-}{-}$ and with equality in place of $\le$).

  Next, observe $m-2 \le n$, so $\ivalrel{}{m-2}{v}{v'}$ and hence
  $\iexprel{}{m-2}{e[v/x]}{e'[v'/x]}$ by the property of $e$ and
  $e'$ above.  Therefore, $\aMeas{m-2}{e[v/x]}{K}{A} \le \Meas{e'[v'/x]}{K'}{A}$.

  Putting the pieces together, we have
  \begin{align*}
    & {\aMeas {m}{\lambdaexp{x}{e}}{K_1}{A}} \\
    & = {\aMeas {m-2} {e[v/x]} {K}{A}} \\
    & \le {\Meas {e'[v'/x]}{K'}{A}} \\
    & = {\Meas {\lambdaexp{x}{e'}}{K_1'}{A}}
  \end{align*}
 \end{proof}


\subsection{From Section~\ref{sec:shuffling}}

\begin{proof}[Proof of \ref{lemma-shuffling-iterate}]
  By Tonelli's Theorem (Lemma~\ref{lemma-tonelli}) and the fact that
  $g$ is arbitrary we can freely rearrange the parameters to $g$
  without loss of generality.  In particular, all paths ending with
  $L$ are assumed to come before any paths ending with $R$.
  
  Let $l$ be the length of the longest path or 0 if $n = 0$. We proceed by strong induction on $l$.
  \begin{itemize}
  \item case $l = 0$. For every $i$, we know that $p_i$ must be the empty list. Since
    we also know that the paths are non-duplicating it follows that $n \leq 1$.
    If $n = 1$ then the equality holds trivially, and if $n = 0$ then by
    Property~\ref{lemma-entropy-splitting}(1) we get
    \begin{displaymath}
      \leb{g()}{\s} = g()
    \end{displaymath}
  \item case $l > 0$.
    Since at least one path is not the empty list, it follows from
    non-duplication that no path is the empty list.  For each $i \in
    \setof{1 \comdots n}$, let $q_i$ be $p_i$ with the last direction
    removed.
  
    Assume without loss of generality that $p_1$ \ldots $p_k$ end
    with $L$ and $p_{k+1}$ \ldots $p_n$ end with $R$, so $p_i(\s) =
    q_i(\pi_L(\s))$ for $1 \le i \le k$ and and $p_i(\s) =
    q_i(\pi_R(\s))$ for $k+1 \le i \le n$.  Then we can conclude:
    \begin{align*}
      &\leb{g(p_1(\s) \comdots p_n(\s))}{\s} \\
      &= \leb{g(q_1(\pi_L(\s)) \comdots q_k(\pi_L(\s)), q_{k+1}(\pi_R(\s)) \comdots q_n(\pi_R(\s)))}{\s} \\
      &= \dint{g(q_1(\s') \comdots q_k(\s'), q_{k+1}(\s'') \comdots q_n(\s''))}{\s'}{\s''} \tag{Property~\ref{lemma-entropy-splitting}(4)} \\
      \intertext{Since every $q_i$ is strictly shorter than $p_i$, we can apply the induction hypothesis, first at $\s'$ and then at $\s''$.}
      &= \leb{\left(\iterint{g(\s_1 \comdots \s_k, q_{k+1}(\s'') \comdots q_n(\s''))}{\s_1}{\s_k}\right)}{\s''} \tag{IH} \\
      &= \iterint{g(\s_1 \comdots \s_n)}{\s_1}{\s_n} \tag{IH}
   \end{align*}
  \end{itemize}
\end{proof}

\subsection{From Section~\ref{sec:commut}}
\label{sec:commut-appendix}

\begin{proof}[Proof of \ref{thm-commut} using Small-Step Semantics]
Let $e$ and $e'$ denote the two expressions of the theorem.
We will use Theorem~\ref{thm-entropy-shuffling1}.  To do so, we will
  consider the evaluations of $e$ and $e'$.   For each
  evaluation, we will use the Interpolation Theorem to define
  waypoints in the evaluation.  We then use the Genericity Theorem to
  establish that the ending configurations are the same.

We begin by watching the first expression evaluate in an arbitrary
continuation $K$, saved entropy $\t$, and initial weight $w$:

\begin{displaymath}
  \wider
  \begin{array}{l@{{}}ll}
     & {\nconfig{\s}{\letexp{x_1}{e_1}{\letexp{x_2}{e_2}{e_0}}}{K}{\t}{w}}
    \\
    \to{} & {\nconfig {\pi_L(\s)} {e_1}
             {\letcont{x_1}{\letexp{x_2}{e_2}{e_0}}{K}}
             {\pi_R(\s) \cons \t} {w}} \\
    \to^*{} & {\nconfig {\s_1} {v_1} {\letcont{x_1}{\letexp{x_2}{e_2}{e_0}}{K}}
                        {\pi_R(\s) \cons \t} {w_1 \times w}}
          &\just{Interpolation} \\
    \to   & {\nconfig {\pi_R(\s)} {\letexp{x_2}{e_2}{e_0[v_1/x_1]}} {K}
              {\t} {w_1 \times w}} \\
    \to   & {\nconfig {\pi_L(\pi_R(\s))} {e_2}
                {\letcont {x_2}{e_0[v_1/x_1]}{K}}
                {\pi_R(\pi_R(\s)) \cons \t} {w_1 \times w}} \\
    \to^*  & {\nconfig {\s_2} {v_2}
                 {\letcont {x_2}{e_0[v_1/x_1]}{K}}
                 {\pi_R(\pi_R(\s)) \cons \t} {w_2 \times w_1 \times
                      w}}    &\just{Interpolation} \\
    \to    & {\nconfig {\pi_R(\pi_R(\s))} {e_0[v_1/x_1][v_2/x_2]}
                 {K} {\t} {w_2 \times w_1 \times w}}
  \end{array}
\end{displaymath}

Next, we outline the analogous computation with the second expression $e'$,
starting in a different entropy $\s'$, but with the same continuation
$K$, saved entropy $\t$, and weight $w$.  We proceed under the
assumption that $e'$ reduces to a value; we will validate this
assumption later.

\begin{displaymath}
  \wider
  \begin{array}{l@{{}}ll}
    & {\nconfig{\s'}{\letexp{x_2}{e_2}{\letexp{x_1}{e_1}{e_0}}}{K}{\t}{w}}
    \\
    \to{} & {\nconfig {\pi_L(\s')} {e_2}
             {\letcont{x_2}{\letexp{x_1}{e_1}{e_0}}{K}}
             {\pi_R(\s') \cons \t} {w}} \\
    \to^*{} & {\nconfig {\s'_2} {v'_2} {\letcont{x_2}{\letexp{x_1}{e_1}{e_0}}{K}}
                        {\pi_R(\s') \cons \t} {w'_2 \times w}}
          &\just{Interpolation} \\
    \to   & {\nconfig {\pi_R(\s')} {\letexp{x_1}{e_1}{e_0[v'_2/x_2]}} {K}
              {\t} {w'_2 \times w}} \\
    \to   & {\nconfig {\pi_L(\pi_R(\s'))} {e_1}
                {\letcont {x_1}{e_0[v'_2/x_2]}{K}}
                {\pi_R(\pi_R(\s')) \cons \t} {w'_2 \times w}} \\
    \to^*  & {\nconfig {\s'_1} {v'_1}
                 {\letcont {x_1}{e_0[v'_2/x_2]}{K}}
                 {\pi_R(\pi_R(\s')) \cons \t} {w'_1 \times w'_2 \times
                      w}}    &\just{Interpolation} \\
    \to    & {\nconfig {\pi_R(\pi_R(\s'))} {e_0[v'_2/x_2][v'_1/x_1]}
                 {K} {\t} {w'_1 \times w'_2 \times w}}
  \end{array}
\end{displaymath}

To get these computations to agree, we choose $\s'$ so that the
entropies for $e_1$, $e_2$ and the substitution instances of $e_0$
are the same in both calculations.   So we choose $\s'$ such that
\begin{displaymath}
  \begin{array}{l@{{}={}}ll}
    \pi_L(\pi_R(\s')) & \pi_L(\s) &\just{entropy for $e_1$} \\
    \pi_L(\s')        & \pi_L(\pi_R(\s))  &\just{entropy for $e_2$} \\
    \pi_R(\pi_R(\s')) &     \pi_R(\pi_R(\s))  &\just{entropy for $e_0$} \\
  \end{array}
\end{displaymath}
This can be accomplished by using $\phi_c$ from
Section~\ref{sec:shuffling}; we set
\begin{displaymath}
  \s' = \phi_c(\s) = \pi_L(\pi_R(\s)) \cons (\pi_L(\s) \cons \pi_R(\pi_R(\s)))
\end{displaymath}

Applying the genericity theorem at $e_1$ we conclude that that $e_1$
reduces to a value at entropy $\pi_L(\pi_R\s')) = \pi_L(\s)$ and
continuation ${\letcont {x_1}{e_0[v'_2/x_2]}{K}}$, that
$v_1 = v_1'$,
and that $w_1 = w_1'$.  Similarly applying the genericity theorem at $e_2$ we
conclude that $e_2$ reduces to a value $v_2' = v_2$ and $w_2 = w_2'$.  So the two calculations
culminate in identical configurations.

So we conclude that
\begin{displaymath}
  {\eval{\s}{e}{K}{\t}1{A}} = {\eval{\phi_c(\s)}{e'}{K}{\t}1{A}}
\end{displaymath}

$\phi_c$ is a non-duplicating FSF, so it is measure-preserving by
Theorem~\ref{thm-shuffling-pushforward}.
Then by Theorem~\ref{thm-entropy-shuffling1}, $\ctxrel{}{}{e}{e'}$.  The
converse holds by symmetry.
\end{proof}

\section{Example Transformation Sketch} \label{sec:appendix-opt}

This section gives a sketch of the transformation of the linear
regression example from Section~\ref{sec:probprog}.

Here is the initial program:

\begin{alltt}
  A = normal(0, 10)
  B = normal(0, 10)
  f(x) = A*x + B
  factor normalpdf(f(2) - 2.4; 0, 1)
  factor normalpdf(f(3) - 2.7; 0, 1)
  factor normalpdf(f(4) - 3.0; 0, 1)
\end{alltt}

\noindent Our first goal is to reshape the program to expose the
conjugacy relationship between \texttt{B}'s prior and the
observations. Specifically, we need the observations to be expressed
with \texttt{B} as the mean. The following steps perform the reshaping.

Inline \ttop{f} using $\letvax$:

\begin{alltt}
  A = normal(0, 10)
  B = normal(0, 10)
  factor normalpdf((A*2 + B) - 2.4; 0, 1)
  factor normalpdf((A*3 + B) - 2.7; 0, 1)
  factor normalpdf((A*4 + B) - 3.0; 0, 1)
\end{alltt}

\noindent Rewrite using $((y + x) - z) = (x - (z - y))$ three
times. (This combines associativity and commutativity of $+$ as well
as other facts relating $+$ and $-$.)

\begin{alltt}
  A = normal(0, 10)
  B = normal(0, 10)
  factor normalpdf(B - (2.4 - A*2); 0, 1)
  factor normalpdf(B - (2.7 - A*3); 0, 1)
  factor normalpdf(B - (3.0 - A*4); 0, 1)
\end{alltt}

\noindent Rewrite using $\ttop{normalpdf}(x-y; 0, s) = \ttop{normalpdf}(y; x, s)$
three times.

\begin{alltt}
  A = normal(0, 10)
  B = normal(0, 10)
  factor normalpdf(2.4 - A*2; B, 1)
  factor normalpdf(2.7 - A*3; B, 1)
  factor normalpdf(3.0 - A*4; B, 1)
\end{alltt}

\noindent Now the conjugacy relationship is explicit. The equation
from Section~\ref{section:quasi-denotational} gives a closed-form
solution to the posterior with respect to a single observation, and it
applies to a specific shape of expression. Our new goal is to
reshape the program so the conjugacy transformation can be applied to
the first observation.

First, apply $\letvax$ in reverse, \emph{after} the first occurrence of
\texttt{B}:

\begin{alltt}
  A = normal(0, 10)
  B = normal(0, 10)
  factor normalpdf(2.4 - A*2; B, 1)
  B1 = B
  factor normalpdf(2.7 - A*3; B1, 1)
  factor normalpdf(3.0 - A*4; B1, 1)
\end{alltt}

\noindent Use $\letsax$ to pull out the argument to the first $\ttop{normalpdf}$:

\begin{alltt}
  A = normal(0, 10)
  err1 = 2.4 - A*2
  B = normal(0, 10)
  factor normalpdf(err1; B, 1)
  B1 = B
  factor normalpdf(2.7 - A*3; B1, 1)
  factor normalpdf(3.0 - A*4; B1, 1)
\end{alltt}

\noindent Apply let-associativity twice (viewing the \texttt{factor}
statement as an implicit let with an unused variable).

\begin{alltt}
  A = normal(0, 10)
  err1 = 2.4 - A*2
  B1 = \{ B = normal(0, 10)
         factor normalpdf(err1; B, 1)
         B \}
  factor normalpdf(2.7 - A*3; B1, 1)
  factor normalpdf(3.0 - A*4; B1, 1)
\end{alltt}

\noindent Now the right-hand side of the \texttt{B1} binding is in the
right shape. We apply the conjugacy rule from
Section~\ref{section:quasi-denotational}:

\begin{alltt}
  A = normal(0, 10)
  err1 = 2.4 - A*2
  B1 = \{ B = normal(\( \left( \frac{1}{10^2} + \frac{1}{1^2} \right)^{(-1)} \left( \frac{0}{10^2} + \frac{\texttt{err1}}{1^2} \right) \), \( \left( \frac{1}{10^2} + \frac{1}{1^2} \right)^{(-1/2)} \))
         factor normalpdf(err1; 0, \( (10^2 + 1^2)^{(1/2)} \))
         B \}
  factor normalpdf(2.7 - A*3; B1, 1)
  factor normalpdf(3.0 - A*4; B1, 1)
\end{alltt}

\noindent We have processed one observation. Now we must clean up and
reset the program so the remaining observations can be processed. (One
of the properties of conjugacy is that the posterior has the same form
as the prior, so observations can be absorbed incrementally.)

We use $\letsax$ to give names to the new mean and scale parameters
for \texttt{B}, let-associativity to ungroup the results, and
$\letvax$ to eliminate the \texttt{B1} binding. Finally, we use
commutativity to move the new \texttt{factor} expression up and out of
the way.

\begin{alltt}
  A = normal(0, 10)
  err1 = 2.4 - A*2
  factor normalpdf(err1; 0, \( (10^2 + 1^2)^{(1/2)} \))
  m1 = \( \left( \frac{1}{10^2} + \frac{1}{1^2} \right)^{(-1)} \left( \frac{0}{10^2} + \frac{\texttt{err1}}{1^2} \right) \)
  s1 = \( \left( \frac{1}{10^2} + \frac{1}{1^2} \right)^{(-1/2)} \)
  B = normal(m1, s1)
  factor normalpdf(2.7 - A*3; B, 1)
  factor normalpdf(3.0 - A*4; B, 1)
\end{alltt}

\noindent That completes the processing of the first observation. In its
place we have \texttt{B} drawn from the posterior distribution (with
respect to that observation) and a \texttt{factor} expression to
score \texttt{A} independent of the choice of \texttt{B}.
We can now repeat the process for the remaining observations.

Alternatively, we could have imported a transformation that used the
closed-form formula for the posterior and normalizer given multiple
observations. That would have led to different shaping steps.

\end{document}
